\documentclass[11pt]{article}
\usepackage[T1]{fontenc}
\usepackage[utf8]{inputenc}
\usepackage[margin=1in]{geometry}

\usepackage{amsfonts,amsthm,amsmath,amssymb}
\usepackage{cite,microtype,graphicx,hyperref}
\usepackage[capitalize,nameinlink]{cleveref}

\newtheorem{theorem}{Theorem}

\newtheorem{lemma}{Lemma}
\newtheorem{corollary}[lemma]{Corollary}
\newtheorem{observation}[lemma]{Observation}

\theoremstyle{definition}

\title{Computational Complexities of Folding}

\author{David Eppstein\thanks{Department of Computer Science, University of California, Irvine. Research supported in part by NSF grant CCF-2212129. A preliminary form of the material in \cref{sec:param} was presented at the 35th Canadian Conference on Computational Geometry (CCCG 2023) \cite{Epp-CCCG-23}. Some results in \cref{sec:counting} were initiated at the 34th Bellairs Winter Workshop on Computational Geometry,
co-organized by Erik Demaine and Godfried Toussaint, on March 22--29, 2019 in Holetown, Barbados, for which we thank the organizers and other participants. This work was presented in preliminary form at OSME 2024 and JCGCG$^3$ 2024. We thank Erik Demaine and Stefan Langermann for helpful discussions at JCGCG$^3$ 2024.}}

\date{ }

\begin{document}
\maketitle  

\begin{abstract}
We prove several hardness results on folding origami crease patterns. Flat-folding finite crease patterns is fixed-parameter tractable in the ply of the folded pattern (how many layers overlap at any point) and the treewidth of an associated cell adjacency graph. Under the exponential time hypothesis, the singly-exponential dependence of our algorithm on treewidth is necessary, even for bounded ply. Improving the dependence on ply would require progress on the unsolved map folding problem. Finding the shape of a polyhedron folded from a net with triangular faces and integer edge lengths is not possible in algebraic computation tree models of computation that at each tree node allow either the computation of arbitrary integer roots of real numbers, or the extraction of roots of polynomials with bounded degree and integer coefficients. For a model of reconfigurable origami with origami squares are attached at one edge by a hinge to a rigid surface, moving from one flat-folded state to another by changing the position of one square at a time is $\mathsf{PSPACE}$-complete, and counting flat-folded states is $\mathsf{\#P}$-complete. For self-similar square crease patterns with infinitely many folds, testing flat-foldability is undecidable.
\end{abstract}

\section{Introduction}

The computer simulation of origami and related folding problems has been well-studied, and several systems can convert a crease pattern into its folded state, or find a folding motion to reach that state~\cite{Tac-OSME-09,GhaDemGer-OSME-18,ZhuFil-PRSA-19,ZhuSchFil-AMR-22,CaiRenDin-AST-17,Kwo-JMD-20,MiyShiYas-JVCA-96}. However, these systems are often heuristic, rather having proven performance guarantees. They may sometimes produce unphysical motions in which paper passes through itself. They may be limited to simple folding steps instead of multi-crease motions. The search space may be pruned heuristically, avoiding a slow combinatorial search but omitting valid solutions~\cite{ZhuFil-PRSA-19}. These methods generally approximate real numbers by floating point, and it can be unclear whether a folded state that appears numerically valid might be mathematically impossible. On the theoretical side, only limited positive results are known, including polynomial time algorithms for folding single-strip patterns with parallel creases~\cite{BerHay-SODA-96,ItoNes-JCDCG3-24}, single-vertex crease patterns~\cite{BerHay-SODA-96}, and map folds of $2\times n$ grids~\cite{Mor-12}. We would like to understand the extent to which improved theory might help make these simulations more robust, and to which theoretical hardness limit our ability to find algorithmic guarantees for these problems.

Past results in this direction include the $\mathsf{NP}$-completeness of finding a flat-folded state of a finite crease pattern~\cite{BerHay-SODA-96,AkiCheDem-JCGCGG-15}, and the undecidability of certain problems related to flat folding of infinite crease patterns~\cite{HullZak-23,Ass-24}. We extend these results by proving the following:
\begin{itemize}
\item Flat folding has a fixed-parameter tractable algorithm, with time depending factorially on the ply of a crease pattern (the maximum number of layers in its folded state), exponentially on the \emph{treewidth} of a graph associated with the crease pattern, and quadratically in the number of creases (\cref{thm:fpt}). Under the exponential time hypothesis, the exponential dependence on treewidth is necessary for this result, even when the ply is bounded~(\cref{thm:eth}).
\item In the radical computation tree and root computation tree models of exact symbolic computation formalized by Bannister et al.~\cite{BanDevEpp-JGAA-15}, it is not always possible to represent the three-dimensional folded state of a polyhedral net with triangular faces and integer edge lengths (\cref{thm:galois}).
\item For a toy model of flat folding consisting of uncreased congruent squares of paper taped to a tabletop by one edge (``flaps and flips''), getting between two flat folded states by moving one square at a time, or testing whether all states are reachable from each other, is $\mathsf{PSPACE}$-complete (\cref{thm:pspace}). Counting flat-folded states of flaps and flips instances is $\mathsf{\#P}$-complete under polynomial-time counting reductions (\cref{thm:count}).
\item There is a fixed infinite self-similar crease pattern (the union of infinitely many scaled copies of a finite set of creases, which may either be labeled or unlabeled) for which testing whether a finite choice of folds can be extended to a flat-folding of the whole pattern is  $\mathsf{coRE}$-complete and undecidable (\cref{thm:undecidable}). Without the finite choice of folds, it is undecidable to determine whether a given self-similar crease pattern is flat-foldable (\cref{thm:wang}).
\end{itemize}

\cref{thm:pspace} and \cref{thm:count} use reductions from nondeterministic constraint logic~\cite{HeaDem-TCS-05}. In proving \cref{thm:count}, we also prove that counting solutions to a nondeterministic constraint logic instance is $\mathsf{\#P}$-complete (\cref{cor:ncl-count}). Few $\mathsf{\#P}$-completeness or $\mathsf{\#P}$-hardness proofs are known for geometric problems~\cite{Lin-SIDMA-86,KamSur-ANALCO-11,DitPak-18,Epp-DCG-20}. \cref{thm:count} adds to this small set. \cref{thm:undecidable} and \cref{thm:wang} strengthen previous results of Hull and Zakharevich~\cite{HullZak-23} by using finite paper, avoiding optional folds, and proving undecidability for flat folding rather than a more complicated decision problem.

\section{Preliminaries: Finite and infinite flat foldings}
\label{sec:flat}

Although many previous works consider flat foldings of origami crease patterns, we need to be somewhat careful with definitions, especially when considering infinite flat foldings. As in previous work~\cite{Epp-JoCG-19}, we base our definition of flat folding on a \emph{local flat folding}, a description of how the folding maps a flat surface to itself without describing the spatial arrangement of layers as a flat-folded surface. \emph{Flat foldings} augment this model to include layer ordering. Define a \emph{crease pattern} to consist of a subset of the plane together with a system of (possibly infinitely many) open line segments, rays, or lines (\emph{creases}), each of which has an open neighborhood within which it is the only crease. It is \emph{finite} when there are finitely many creases.
Define a \emph{local flat folding} of a crease pattern $P$ to be a continuous function $\varphi$ such that:
\begin{itemize}
\item For each open subset $S$ of $P$ that does not intersect any crease, $\varphi$ is an isometry.
\item For each open subset $S$ of $P$ that contains a single crease $c$, any two points that are reflections of each other across $c$ have the same image.
\end{itemize}
In short, $\varphi$ is a \emph{continuous piecewise isometry} with the creases as boundaries of the pieces on which it acts isometrically. Define a \emph{vertex} of a crease pattern to be a point at the endpoint of more than one crease. (A crease may also have an endpoint at the boundary of the crease pattern, which we do not consider as a vertex.) This level of detail does not distinguish mountain folds from valley folds.

\begin{observation}
\label{obs:local}
Given a finite crease pattern $P$, we can in linear time construct a local flat folding $\varphi$ for $P$, if one exists, or determine that it does not exist.
\end{observation}

\begin{proof}
Choose an arbitrary starting polygon, set $\varphi$ to be the identity within this polygon, and traverse the adjacencies between polygons of the decomposition. When traversing an edge from a polygon whose mapping under $\varphi$ has been determined to a polygon whose mapping has not, set the mapping for the new polygon to be the reflection of the mapping for the old polygon across the traversed edge. When traversing an edge to a polygon whose mapping has already been determined, check its consistency with this reflection.
\end{proof}

The function $\varphi$, constructed in this way, is unique up to rigid transformations of the plane.

\begin{figure}[t]
\centering\includegraphics[width=0.4\textwidth]{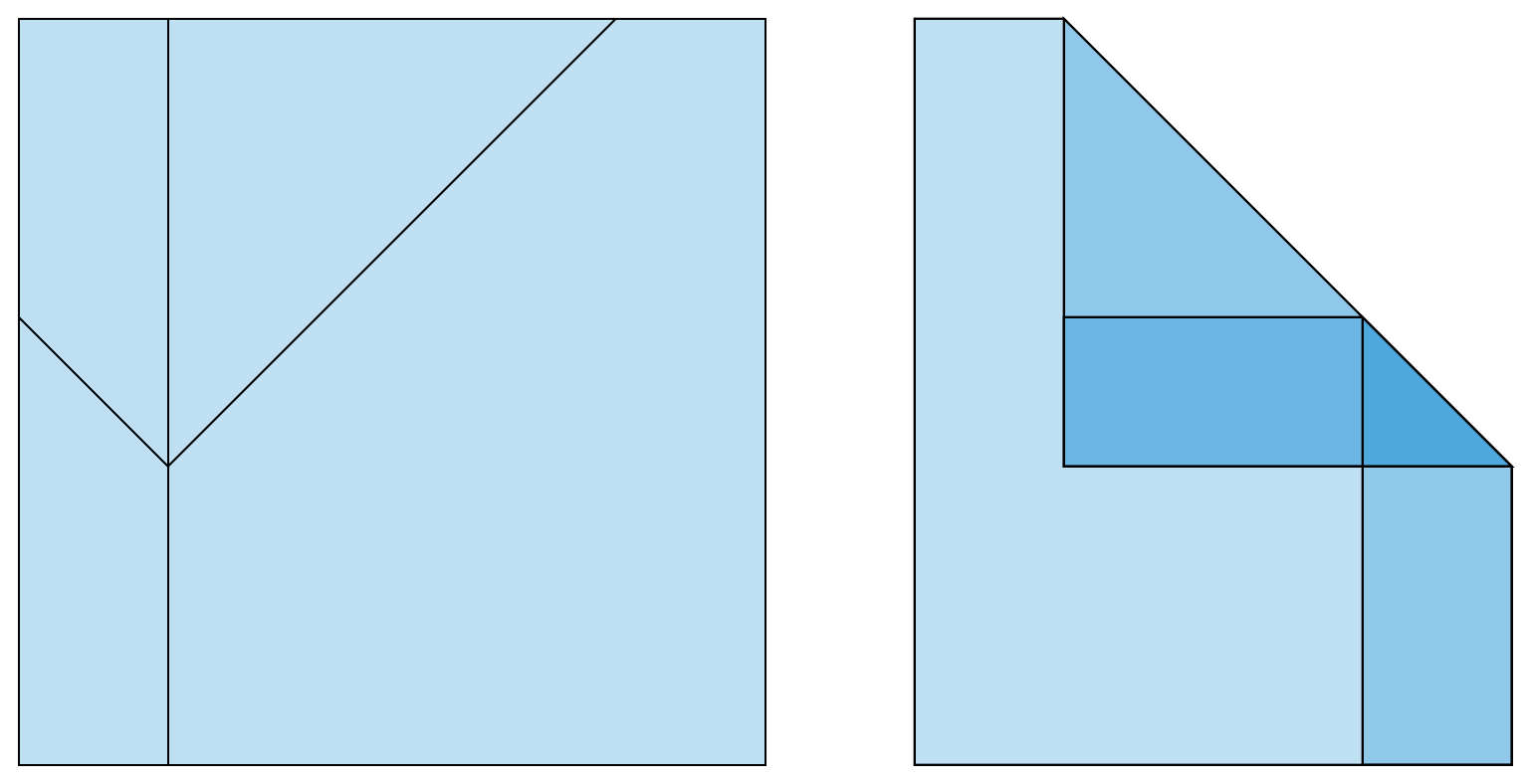}
\caption{A finite crease pattern (left) and the arrangement of its local flat folding (right), with shading indicating the ply of each arrangement cell. The ply of the overall pattern is four, equal to the maximum ply in the small triangular cell.}
\label{fig:arrangement-ply}
\end{figure}

Define the \emph{arrangement} of a local flat folding $\varphi$ of a crease pattern $P$ to be the system of images under $\varphi$ of the creases in $P$. These images of creases partition the plane into \emph{cells}, maximal regions that are not crossed by the image of any crease. It is convenient to consider preimages of  cells rather than of individual points. The \emph{ply} of a cell is the number of these preimages, and the ply of the crease pattern is the maximum ply of any cell (possibly infinite for infinite crease patterns). See \cref{fig:arrangement-ply}. Using standard methods in computational geometry, an arrangement of a local flat folding with $n$ creases has complexity $O(n^2)$. It can be constructed and its ply calculated in time $O(n^2)$.

Define a \emph{layering} of a local flat folding to be a vertical ordering of the preimages for each cell of the arrangement. Define a \emph{flat folding} to be a local flat folding and a layering that, for every $\varepsilon>0$, is consistent with a topological embedding of the crease pattern into three-dimensional space that is $\varepsilon$-close to the local flat folding. Here, ``close'' means there exists a local flat folding into a plane in space so that, for every point of the crease pattern, its images under the topological embedding and under the local flat folding have distance at most $\varepsilon$ from each other. To avoid topological difficulties we additionally require that a line perpendicular to the plane, through a point of the plane farther than $\varepsilon$ from any crease, has exactly one point of intersection with each preimage of a cell: the embedding cannot be ``crumpled'' far from its creases. This restriction gives us a vertical ordering for each cell, the intersection order along any such perpendicular line.

A cross-section of such a topological embedding across any crease includes layers in two adjacent arrangement cells. Two layers in the same cell can pair up to form a crease, two layers from each cell can be paired to form parts of an uncreased region, and the boundary of the crease pattern may lie along a crease (\cref{fig:layer-consistency}, left). These paired layers must meet certain obvious conditions:
\begin{itemize}
\item If two polygons span the two cells without being creased, they must be consistently ordered in both cells instead of crossing at the crease (\cref{fig:layer-consistency}, top right).
\item If two layers of the same cell meet in a crease, any uncreased polygon spanning the two cells without being creased cannot lie between the two creased layers, as their crease would block it  (\cref{fig:layer-consistency}, middle right). This is the \emph{taco-tortilla property} of Akitaya et al.~\cite{AkiCheDem-JCGCGG-15}.
\item If two creased pairs of layers form the same crease line, their layers cannot alternate, as this would form a crossing (\cref{fig:layer-consistency}, bottom right). This is the \emph{taco-taco property}~\cite{AkiCheDem-JCGCGG-15}.
\item If two layers of the same cell meet in a crease labeled as being a mountain fold or valley fold, the ordering of the layers must be consistent with the fold type (not shown).
\end{itemize}
Define a layering to be \emph{uncrossed} when, at each crease, it meets all of these conditions.

\begin{lemma}
\label{lem:uncrossed}
A local flat folding of a finite crease pattern comes from a flat folding if and only if it has an uncrossed layering.
\end{lemma}

\begin{proof}
In one direction, if a flat folding exists, it cannot violate any of the conditions above, because each describes a certain type of crossing, and topological embeddings forbid crossings. In the other direction, every uncrossed layering comes from a flat folding: one can form a 3d embedding from it, by shrinking each cell a small distance from its boundary, making parallel copies of the cell in 3d in the order given by the layering, all separated from each other but within distance $\varepsilon$ of the plane of the local flat folding, and connecting them with curved patches of surface near each crease.

\begin{figure}[t]
\centering\includegraphics[width=0.6\textwidth]{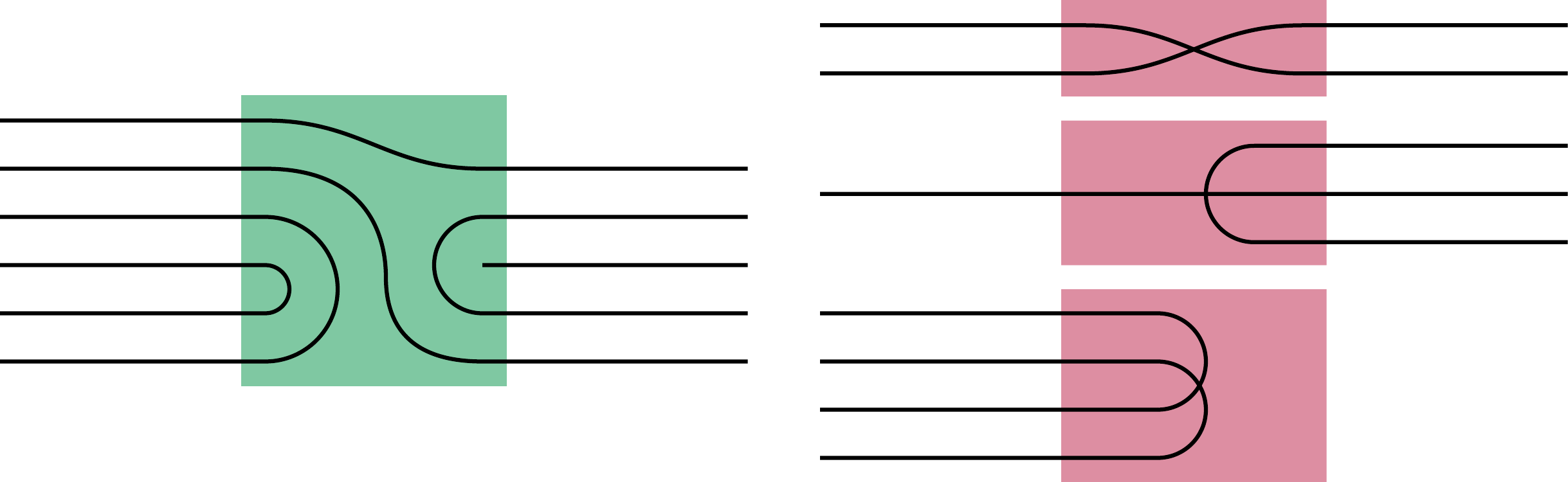}
\caption{Left: cross-section through a crease (shaded region) of a uncrossed layering. Right: Three ways that a layering can be inconsistent across a crease: two uncreased polygons cross (top), an uncreased polygon is blocked by two layers that connect to form a crease (middle), or two pairs of creased layers cross (bottom).}
\label{fig:layer-consistency}
\end{figure}

It is unnecessary to add more case analysis for the way layerings can interact at a vertex, instead of across a crease. Two surfaces in 3d cannot cross each other at a single point, without crossing along a curve touching that point, so if a system of surfaces in 3d defined from a uncrossed layering avoids crossings except at points $\varepsilon$-close to the vertices, it can be converted into a topological embedding for the same layering that avoids crossing everywhere.
\end{proof}

\section{Parameterized and fine-grained complexity}
\label{sec:param}

In this section, we show that flat-foldability is \emph{fixed-parameter tractable} when parameterized by two values: the \emph{ply} of the crease pattern (how many layers of paper can overlap at any point of the flat-folded result), and the \emph{treewidth} of an associated \emph{cell adjacency graph} constructed by overlaying the flat polygons of the crease pattern in the positions they would take in their folded state. The pattern may either be labeled with mountain and valley folds or unlabeled. We identify a wide class of patterns for which flat foldability is easy: those with bounded ply and bounded treewidth.

Bounded ply is natural in paper folding, as large ply can make physical realization difficult~\cite{DemEppHes-JDA-16}. Treewidth generalizes the notion of a crease pattern that is complicated only in one dimension, and simple in a perpendicular dimension, as occurs (with large ply) for $2\times n$ map folding.  Single-vertex crease patterns automatically have low treewidth (their cell adjacency graph is just a cycle; see \cref{sec:cell-adj}) but may have high ply. Fixed-parameter tractability means that the worst-case time bound has the form of a polynomial in the input size, multiplied by a non-polynomial function of the parameters; in our case this function is factorial in the ply and exponential in the treewidth. On inputs with bounded parameters, the time bound simplifies to a polynomial of the input size.

 For flat foldings of bounded ply, our algorithm is single-exponential in the treewidth. This dependence is unavoidable under the exponential time hypothesis (\cref{thm:eth}). We do not have as strong a justification for the dependence on ply, but if it could be eliminated, we could solve \emph{map folding} in polynomial time. Map folding concerns flat-foldability of a square grid with edges labeled as mountain and valley folds. It has a polynomial time algorithm for $1\times n$ grids~\cite{ArkBenDem-CGTA-04} and for $2\times n$ grids~\cite{Mor-12}, but for larger grids its complexity is a major open problem. The associated cell adjacency graph is trivial and has bounded treewidth. If its time turns out to be non-polynomial, that non-polynomial dependence must come from the only other parameter, ply.

\begin{figure}[t]
\includegraphics[width=0.7\textwidth]{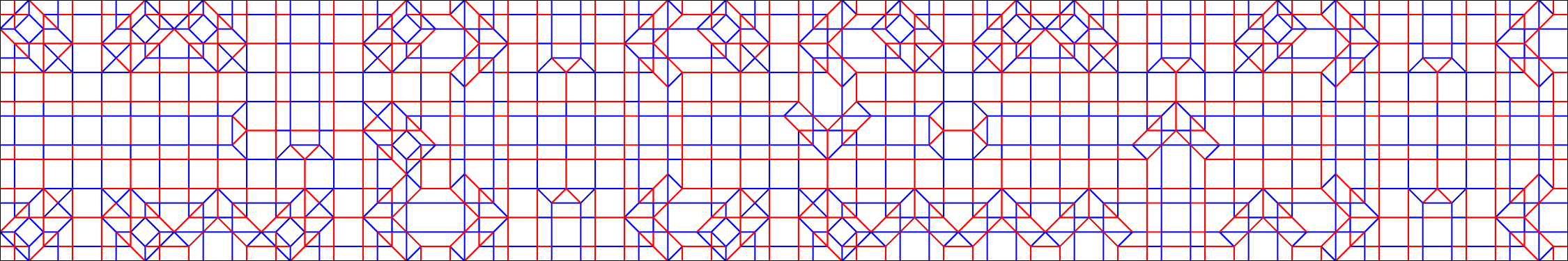}\hfill \includegraphics[width=0.25\textwidth]{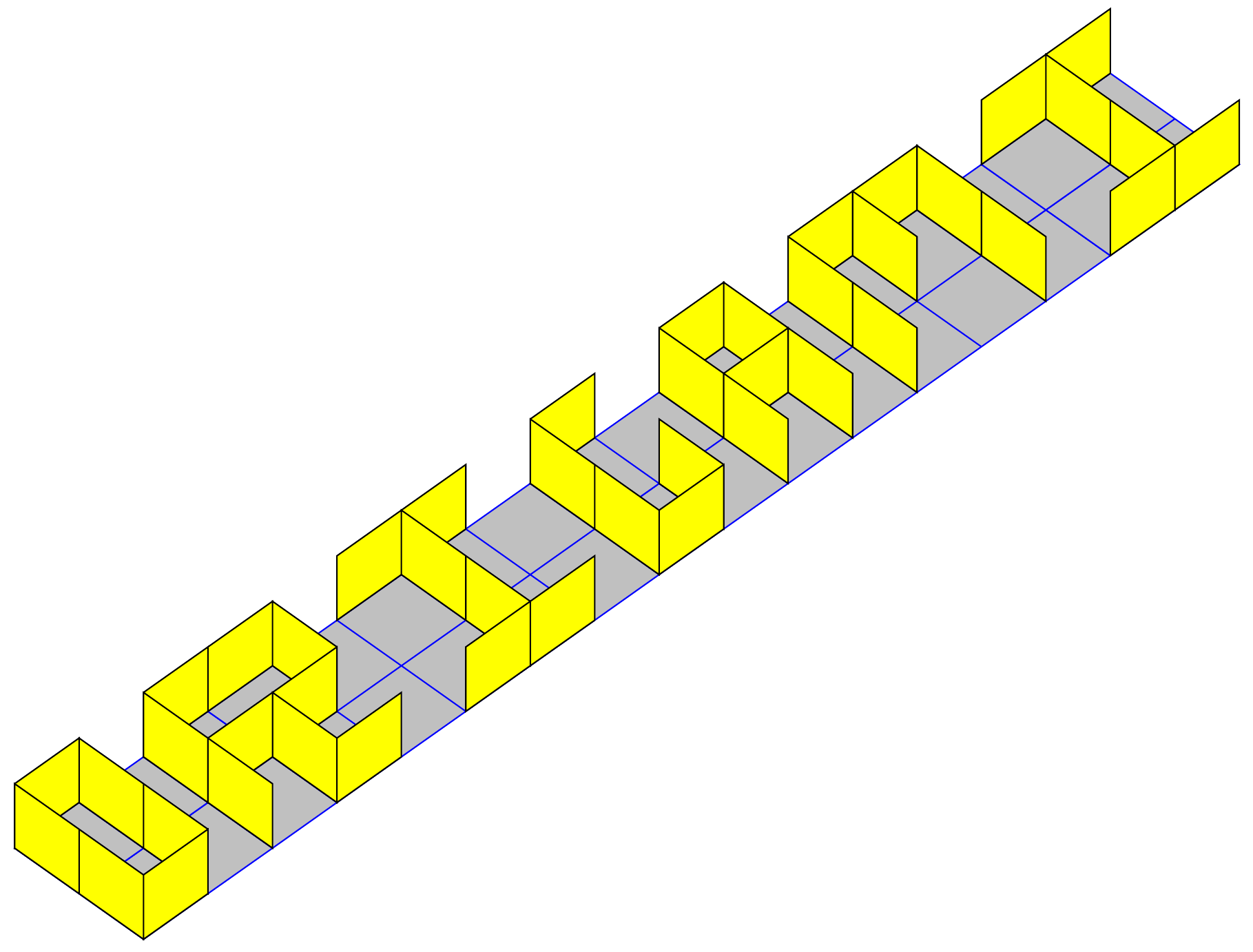}
\caption{A crease pattern for the origami font of Demaine, Demaine, and Ku, produced by \url{http://erikdemaine.org/fonts/maze/?text=origami}, and its 3d folded form.}
\label{fig:origami-maze}
\end{figure}

The parameters of our algorithm are naturally bounded in crease patterns from the origami font of Demaine, Demaine, and Ku~\cite{DemDemKu-G4G-10,DemDem-TCS-15,DemDem-JIP-20}. Rendering text in this font converts it into an origami crease pattern. When folded, this pattern produces a three-dimensional structure consisting of letterform-shaped vertical walls on a flat background surface (\cref{fig:origami-maze}). The resulting structures are not flat foldings because of the vertical walls, but can easily be modified to be. The resulting crease pattern, for a line of text, has bounded ply, high complexity along any horizontal line through the pattern, and low complexity along any vertical line. Its cell adjacency graph has bounded bandwidth, and for a modified version of the font that included ascenders and descenders it would instead have bounded pathwidth, special cases of our bounded treewidth assumption.

\subsection{Cell adjacency graphs and their treewidth}
\label{sec:treewidth}

A \emph{tree decomposition} of a graph $G$ consists of an unrooted tree $T$, and an assignment to each tree vertex $t_i$ of a set $B_i$ of vertices from $G$ (called a \emph{bag}), such that each vertex of $G$ belongs to the bags from a connected subtree of $T$, and each edge of $G$ has endpoints that belong together in at least one bag. Its \emph{width} is the maximum size of a bag, minus one, and the \emph{treewidth} of $G$ is the minimum width of any tree decomposition. Many hard graph optimization problems can be solved in linear time on graphs of bounded treewidth, using dynamic programming over tree decompositions. Finding the treewidth is hard but can be solved in linear time for graphs of bounded treewidth~\cite{Bod-SICOMP-96}. Our application uses treewidth of planar graphs, derived from the arrangement of a crease pattern. It is unknown whether planar treewidth is hard, but it has an (unparameterized) polynomial time approximation ratio of $3/2$ via an algorithm for a closely related parameter, \emph{branchwidth}~\cite{SeyTho-Comb-94}.

To simplify our algorithm we use  \emph{nice tree decompositions}, rooted trees with four types of bags:
\begin{itemize}
\item \emph{Leaf bags}, leaves of the rooted tree, have exactly one graph vertex in the bag.
\item \emph{Introduce bags} have exactly one child vertex in the tree, and their bag differs from that of the child by the addition of exactly one graph vertex.
\item \emph{Forget bags} have exactly one child vertex in the tree, and their bag differs from that of the child by the removal of exactly one graph vertex.
\item \emph{Join bags} have exactly two children, whose bags are both equal to the join bag.
\end{itemize}

A nice tree decomposition can be constructed in linear time from an arbitrary tree decomposition, without increasing the width, and it has size linear in the size of the input tree decomposition~\cite{Klo-94}.

\label{sec:cell-adj}

\begin{figure}[t]
\centering\includegraphics[scale=0.3]{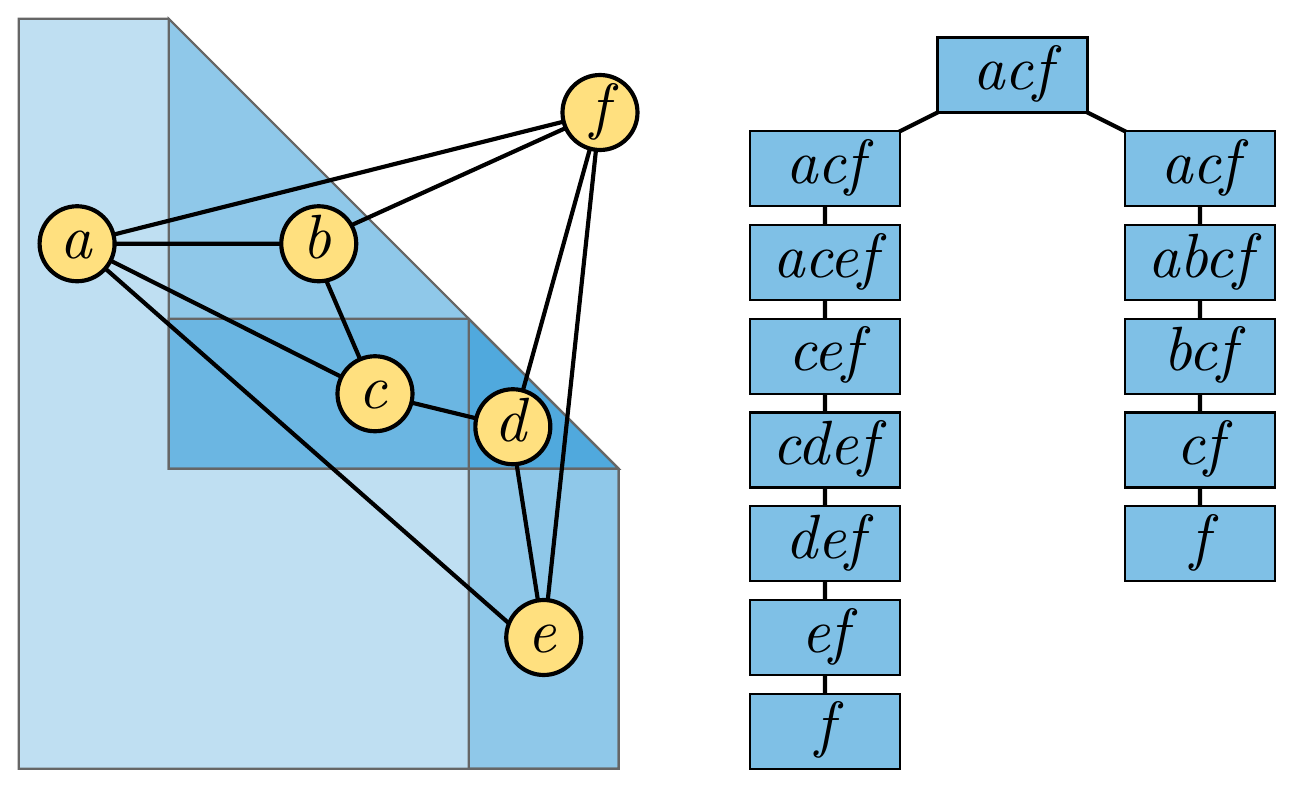}
\caption{The arrangement from \cref{fig:arrangement-ply}, its cell adjacency graph, and a nice tree decomposition.}
\label{fig:cell-adj-decomp}
\end{figure}

For finite crease patterns, our definition of flat folding involves constructing an arrangement of polygons, the images of the polygons in the crease pattern under the mapping that defines a local flat folding. Define its \emph{cell adjacency graph} to be the graph with a vertex for each cell of the arrangement, and an edge between each two neighboring cells~\cite{DenShu-DCG-88}. We include even the  cells with ply zero, in order to check for crossings along the creases between this cell and its neighbors. For example, in map folding, a square grid crease pattern is folded down to a single square, but the arrangement has two cells, the inside of the square and the outside, so the cell adjacency graph is $K_2$. In the case of a single-vertex crease pattern, the local flat folding produces an arrangement consisting of wedges all having this vertex as their apex, and its cell adjacency graph is a cycle.

The two main parameters for the analysis of our algorithm will be the ply of the local flat folding, and the treewidth of the cell adjacency graph. \cref{fig:cell-adj-decomp} depicts an example of a cell adjacency graph of treewidth 2, and a nice tree decomposition with a join bag at its root.

\subsection{The algorithm}

We will test the flat foldability of a crease pattern by first attempting to construct its local flat folding. If this step fails, a flat folding does not exist, and our algorithm exits. Next, we construct its arrangement and its cell adjacency graph,  find an optimal or near-optimal tree decomposition of the cell adjacency graph, and convert the tree decomposition to a nice tree decomposition of the same width. Finally, we reach the main part of our algorithm: a bottom-up dynamic program on the bags of the tree decomposition. If $B$ is any bag (that is, a set of cells of the arrangement, associated with a vertex of the nice tree decomposition), we define a \emph{state} of $B$ to be a layering of each cell in $B$.

\begin{observation}
In a tree decomposition of width $w$ for a crease pattern of ply $p$, every bag has at most $(p!)^{w+1}$ states.
\end{observation}

If $B$ has a child $C$ in the tree decomposition,
then we say that a state of $B$ is \emph{consistent} with a state of $C$ if they have the same layering in all cells that belong to both bags. We say that a state of  $B$ is \emph{locally uncrossed} if, for all pairs of adjacent cells in $B$, their layerings meet the same conditions that we used earlier to define a global layering as being uncrossed. We say that a state is \emph{valid} when it is locally uncrossed and is consistent with (recursively defined) valid states for all child bags.

\begin{lemma}
\label{lem:valid-fold}
A bag $B$ has a valid state if and only if there exists a layering that is uncrossed at all creases between pairs of cells that occur together in $B$ or its descendants in the tree decomposition.
\end{lemma}

\begin{proof}
If such a layering exists, its restriction to the cells in $B$ and its descendant bags produces a valid state.
If a valid state exists, coming from a recursively constructed set of valid states among its descendant bags, then each of these states must consistently layer the cells that they have in common, by the requirement of tree-decompositions that each vertex belong to bags in a connected subtree. Form a global layering by choosing arbitrarily a layering for each cell that is not included among these descendants. Then it must be uncrossed at all creases between pairs of cells that occur together in $B$ or its descendants, because any crossing would cause the state to be invalid at that bag, violating the assumption that we have a recursively constructed set of valid states.
\end{proof}

\begin{lemma}
\label{lem:compute-valid}
If we have already computed the valid states of each child of a given bag $B$ of a nice tree decomposition, we can compute the valid states for $B$ itself in time $O(pw(p!)^{w+1})$.
\end{lemma}

\begin{proof} 
We apply a case analysis according to the type of $B$ in the decomposition.
\begin{itemize}
\item At a leaf bag, all states are valid: there are no creases between pairs of cells to cause crossings.
\item At an introduce bag, we add a layering for the introduced cell to all valid layerings of the other cells from the child node. For each child layering, and each layering of the introduced cell, we check at most $w$ previously-unrepresented creases, each in time $O(p)$, to determine whether it forms any forbidden crossing type.
\item At a forget bag, all valid states of the child node determine a valid state of the bag, by forgetting the layering on the cell that is not included.
\item At a join bag, a state is valid when it is valid in both children. We can intersect the valid states in both children, in time linear in the possible states, using a bit array.\qedhere
\end{itemize}
\end{proof}

Putting these pieces together gives our main result:

\begin{theorem}
\label{thm:fpt}
Testing flat foldability of a crease pattern with $n$ creases and ply $p$, with a cell adjacency graph of treewidth $w$, can be performed in time $(p!)^{O(w)}n^2$. 
\end{theorem}

\begin{proof}
We construct a nice tree decomposition and traverse it in bottom-to-top order, using \cref{lem:compute-valid} to determine the valid states in each bag. A folding exists if and only if there is a valid state at the root bag, by  \cref{lem:valid-fold}.
The arrangement of the local flat folding has size $O(n^2)$ and can be constructed in time $O(n^2)$ using standard algorithms from computational geometry, giving also the cell adjacency graph. From this graph, a tree decomposition of width $O(w)$ can be constructed in time $2^{O(w)}n^2$ using recently developed parameterized algorithms for approximating treewidth~\cite{BodDraDre-SICOMP-16,Kor-FOCS-21}.
The quadratic dependence on $n$ comes from the size of the arrangement of the local flat folding, and the size of the tree decomposition of its cell adjacency graph. The dependence on ply and width comes from the time bound per bag in \cref{lem:compute-valid}, applied to the width of the constructed tree decomposition (larger by a constant factor than the width of the cell adjacency graph). The $pw$ term in the bound of \cref{lem:compute-valid} is subsumed by other factors in the stated time bound.
\end{proof}

\subsection{ETH-hardness}

The time bound for our parameterized algorithm for flat folding in \cref{thm:fpt}, for crease patterns with ply $O(1)$, simplifies to $2^{O(w)}n^2$, single-exponential in the treewidth of the cell adjacency graph. As we now show, a bound of this form is necessary under the exponential-time hypothesis~\cite{ImpPatZan-JCSS-01}, which for our purposes is most conveniently phrased as the assumption that there does not exist an algorithm for the 3SAT (satisfiability of 3-CNF Boolean formulae with $n$ variables and $m$ clauses) that has a subexponential running time bound of the form $2^{o(n+m)}$. Our proof uses NAE3SAT (not-all-equal-3-satisfiability), a variant of 3SAT. An NAE3SAT instance consists of a system of $n$ Boolean variables and $m$ clauses, where each clause is formed by a triple of variables or their negations. It is satisfied by a truth assignment to the variables such that the three Boolean values in every clause are not all equal. Standard reductions transform any instance of 3SAT to an equivalent instance of NAE3SAT, multiplying both the number of variables and clauses by $O(1)$ factors. It follows that, under the exponential time hypothesis, it is not possible to solve NAE3SAT instances in time subexponential in their numbers of variables or clauses. The same is true more generally for a wide class of satisfiability problems including both 3SAT and NAE3SAT~\cite{JonLagNor-SODA-13}.

We base our hardness result on the proof  by Bern and Hayes that flat foldability is NP-complete~\cite{BerHay-SODA-96}, as corrected by Akitaya et al.~\cite{AkiCheDem-JCGCGG-15}. These works provide two proofs for unlabeled crease patterns and for crease patterns labeled with mountain and valley folds, both following the same outline. They are reductions from NAE3SAT that produce crease patterns in the shape of a rectangle. Each NAE3SAT variable is represented by closely-spaced horizontal crease lines that can be folded consistently in two ways representing true and false values. Each NAE3SAT clause is represented by a \emph{clause gadget} near the top of the rectangle, connected to its three variables by vertical systems of parallel creases that cross the other variables without otherwise interacting with them. The clause gadgets can only be flat-folded for truth assignments that assign unequal values to the connected variables. When a flat folding exists, and the construction is flat-folded, most of the paper has ply 1, with bounded ply within the various gadgets of the construction. \cref{fig:bern-hayes} provides a schematic view.

\begin{figure}[t]
\centering\includegraphics[width=0.8\textwidth]{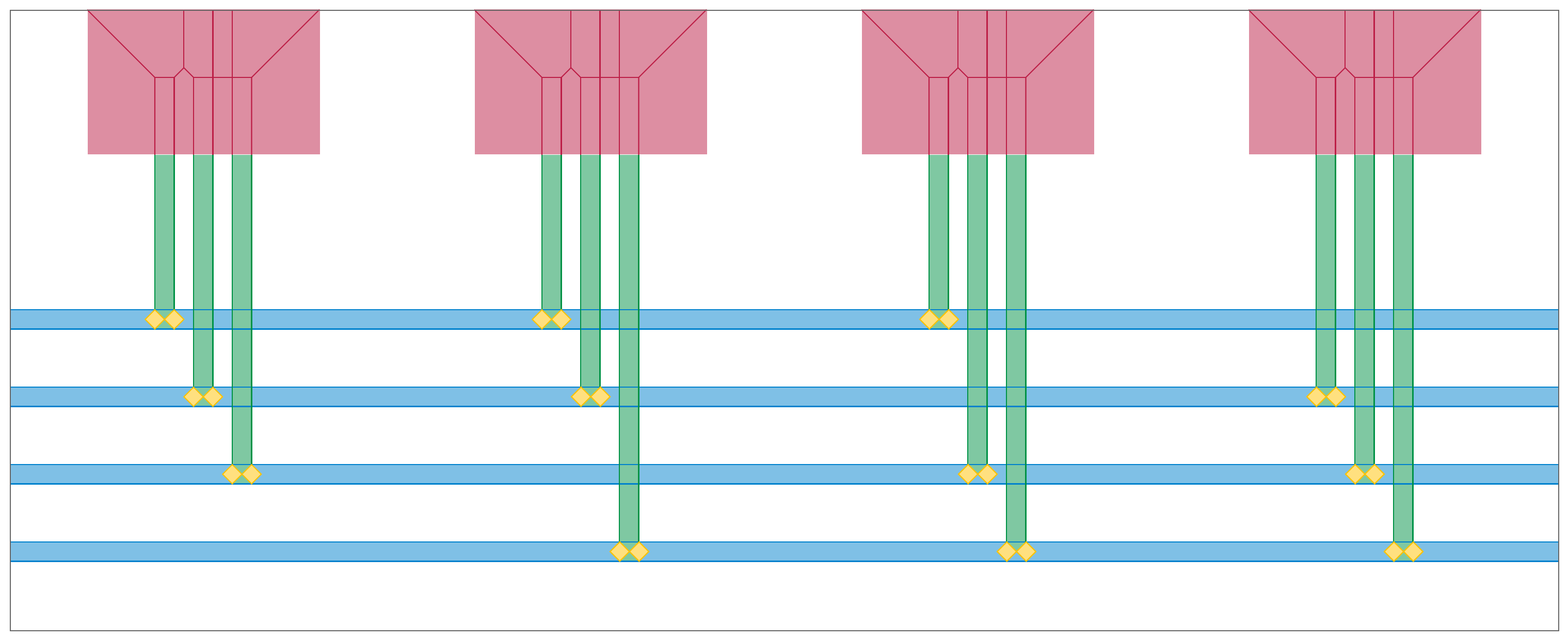}
\caption{Schematic view of the hardness reductions of Bern and Hayes~\cite{BerHay-SODA-96}, as corrected by Akitaya et al.~\cite{AkiCheDem-JCGCGG-15}. The red regions at top are clause gadgets and the blue horizontal left--right paths are variable gadgets. Variable gadgets are connected to clause gadgets by vertical signals (light green). Yellow splitter gadgets connect variables to signals; crossover gadgets are depicted as green squares.}
\label{fig:bern-hayes}
\end{figure}

\begin{observation}
The local flat foldings of the crease patterns of Bern and Hayes and Akitaya et al. have ply $O(1)$. For a NAE3SAT instance with $n$ vertices and $m$ clauses, they have treewidth $O(n)$, obtained by a path decomposition whose bags are the subsets of cells of the local flat folding intersected by vertical lines, in left-to-right order.
\end{observation}

\begin{theorem}
\label{thm:eth}
If the exponential time hypothesis is true, it is not possible to test flat foldability of crease patterns of ply $O(1)$ and treewidth $w$ in time $2^{o(w)}$, regardless of whether the pattern is labeled with mountain and valley folds or unlabeled.
\end{theorem}

\begin{proof}
If such a fast test existed, then applying it to the crease patterns produced by the hardness reductions of Bern and Hayes and Akitaya et al. would give an algorithm for NAE3SAT with time $2^{o(m)}$, contradicting the exponential time hypothesis.
\end{proof}

\section{Galois complexity}
\label{sec:galois}

For flat-foldable and finite crease patterns, all vertex positions in the final folded pattern are reflections of their initial positions across finitely many creases (\cref{obs:local}). It follows that, if the crease endpoints have rational Cartesian coordinates, then all final vertex positions will also be rational. In contrast, as we show in this section, for fully three-dimensional folding, even the accurate representation of the final folded state can be computationally difficult.

In previous work with Bannister et al.~\cite{BanDevEpp-JGAA-15}, we formalized problems of the exact representation of numeric values arising from geometric computation under the framework of \emph{Galois complexity}. Many geometric problems with integer or combinatorial inputs have output coordinate values that are \emph{algebraic numbers}, roots of polynomials with integer coefficients. The exact representation of these numbers as closed-form expressions using radicals has been studied in \emph{Galois theory}, according to which groups associated with these polynomials (their \emph{Galois groups}) determine the existence of a closed-form formula. The algebraic numbers that one wishes to represent have a closed-form formula if and only if the group is a \emph{solvable group}, having a decomposition into a sequence of cyclic quotient groups. In the \emph{radical computation tree} model of computation of Bannister et al., based on algebraic computation trees, with arithmetic and radicals as its fundamental operations, the roots of polynomials with unsolvable Galois groups cannot be computed. An alternative model of computation of Bannister et al. is the \emph{root computation tree}. This uses algebraic computation trees with arithmetic and roots of bounded-degree polynomials as its fundamental operations. As Bannister et al. proved, this model is incapable of constructing high-order regular polygons:
 there exist infinitely many prime numbers $p$ such that constructing a regular $p$-gon requires taking roots of polynomials of unbounded degree, $\Omega(p^{0.677})$.

\begin{figure}[t]
\centering\includegraphics[width=0.4\textwidth]{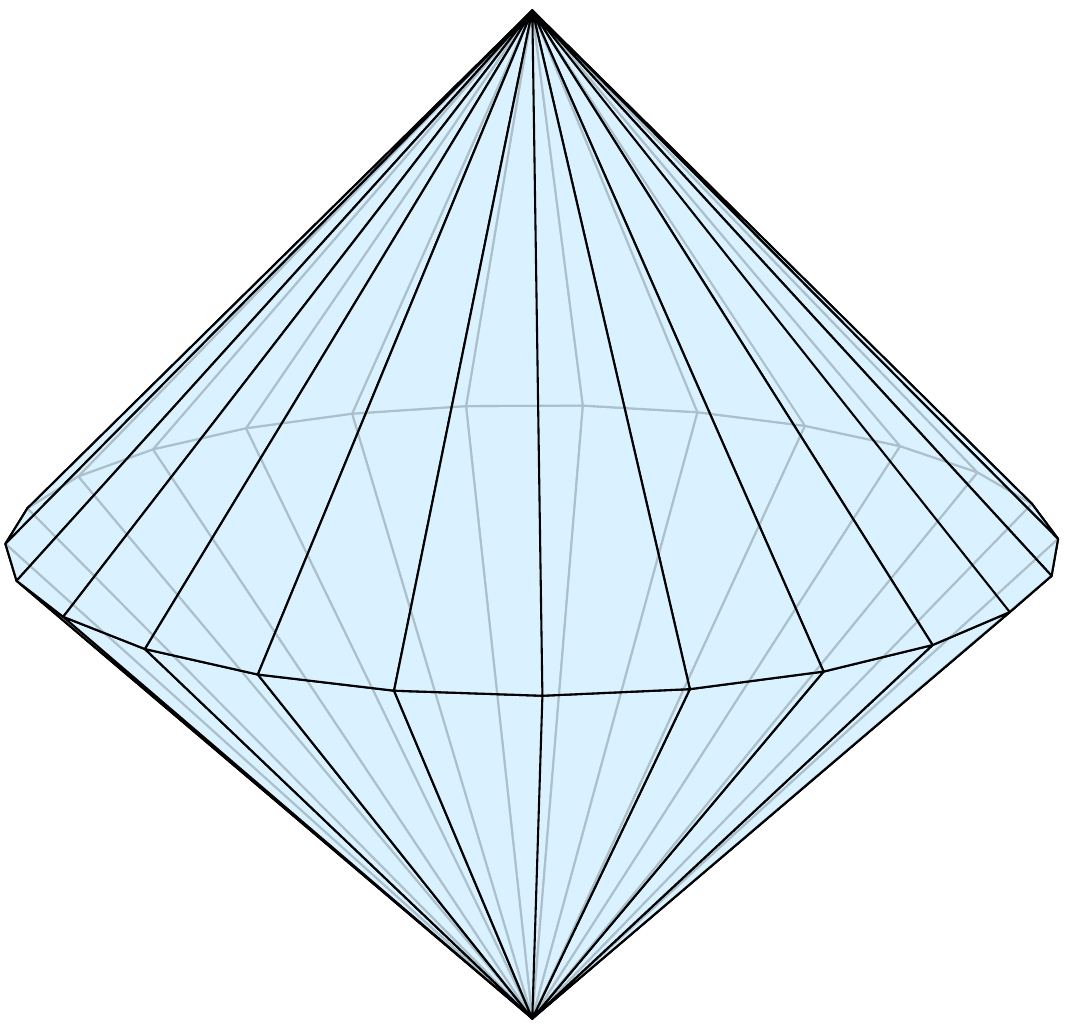}
\caption{A bipyramid with congruent isosceles faces and a regular polygon as its equator, from \cite{Epp-GCOM-21}.}
\label{fig:bipyramid}
\end{figure}

To apply this theory to folding, we need patterns whose folded states are easy to describe, but have coordinates that are hard to represent. The folded states that we use are \emph{bipyramids}, convex polyhedra in which the faces are isosceles triangles, with their bases forming a cycle (the \emph{equator}) and their apexes meeting at two \emph{poles} (\cref{fig:bipyramid}). We choose each face to have integer edge lengths, with a common length $\ell$ used for the two equal side lengths of all face isosceles triangles. A triangle with given integer edge lengths is easy to construct with compass and straightedge (given a unit segment), and all compass-and-straightedge constructions can be performed by both the radical computation tree and root computation tree models~\cite{BanDevEpp-JGAA-15}. Thus, there is no representational difficulty in the input to the folding problem, a crease pattern for a given bipyramid. The output is algebraic, satisfying certain quadratic equations constraining each polyhedron edge to have its specified length. The question is how easy these algebraic numbers are to represent.

We need the following preliminary result on the polyhedral realizations of bipyramids.

\begin{lemma}[Alexandrov uniqueness]
Any combinatorial bipyramid,  with specified edge lengths that are all equal for edges incident to bipyramid poles, is realizable as a convex polyhedron if the sum of apex angles at a pole is less than $2\pi$. If it is realizable, its realization is unique up to congruence. Its equator lies in a plane, with the vertices of the equator on a circle in the plane.
\end{lemma}

\begin{proof}
Existence and uniqueness of the realization are a special case of Alexandrov's uniqueness theorem~\cite{Ale-06}, according to which any edge-to-edge gluing of a polyhedron with the topology of a sphere and with positive angular deficit at each vertex (necessarily summing to $4\pi$ by the Gauss--Bonnet theorem) has a unique realization as a convex polyhedron.  At a vertex of the equator of the bipyramid, four acute base angles of isosceles triangles vertices meet, summing to less than $2\pi$ and giving positive deficit. The positive deficit at the poles is an assumption of our lemma. Therefore, the preconditions for Alexandrov uniqueness are met.

By uniqueness, every combinatorial symmetry of the gluing pattern extends to a geometric symmetry of the realization, for otherwise the realization of a combinatorially symmetric copy of the gluing pattern would give a second realization. Thus, the symmetry that swaps the two poles of the bipyramid (combinatorially) and fixes the equatorial vertices is realized by a geometric symmetry that likewise swaps the poles and fixes the equator. This can only be a reflection symmetry, with the equator lying on the plane of reflection. The two poles must be reflections of each other, so the line segment between the two poles crosses this reflection plane perpendicularly.

All edges incident to poles have the same length $\ell$. Let $s$ be half the length of the line segment between the two poles (the \emph{half-diagonal} of the bipyramid), and let $m$ be the midpoint of this segment, at distance $s$ from each pole. Then each equator vertex forms a right triangle with $m$ and each pole, because the segment from the pole to $s$ is perpendicular to the reflection plane containing the equator vertex. By the Pythagorean theorem, the distance from the equator vertex to $m$ is $\sqrt{\ell^2-s^2}$. Because this distance formula is the same for all equator vertices, they are equidistance from $m$, and lie in a circle on the reflection plane with this radius, centered at $m$.
\end{proof}

\begin{observation}
For every cyclic polygon containing the center of its circle and for every $\ell$ greater than the radius of the polygon, there exists a convex bipyramid with the given polygon as its equator and with all edges incident to poles having length $\ell$.
\end{observation}

\begin{proof}
Use the Pythagorean theorem in the same way to determine the half-diagonal $s$ as $\sqrt{\ell^2-r^2}$, where $r$ is the radius of the circle, and place the two poles at distance $s$ along a line perpendicular to the plane of the polygon through the center of the circle.
\end{proof}

Because of these result, from now on we can consider cyclic polygons in the plane instead of bipyramids in space.

\begin{figure}[t]
\centering\includegraphics[width=0.35\textwidth]{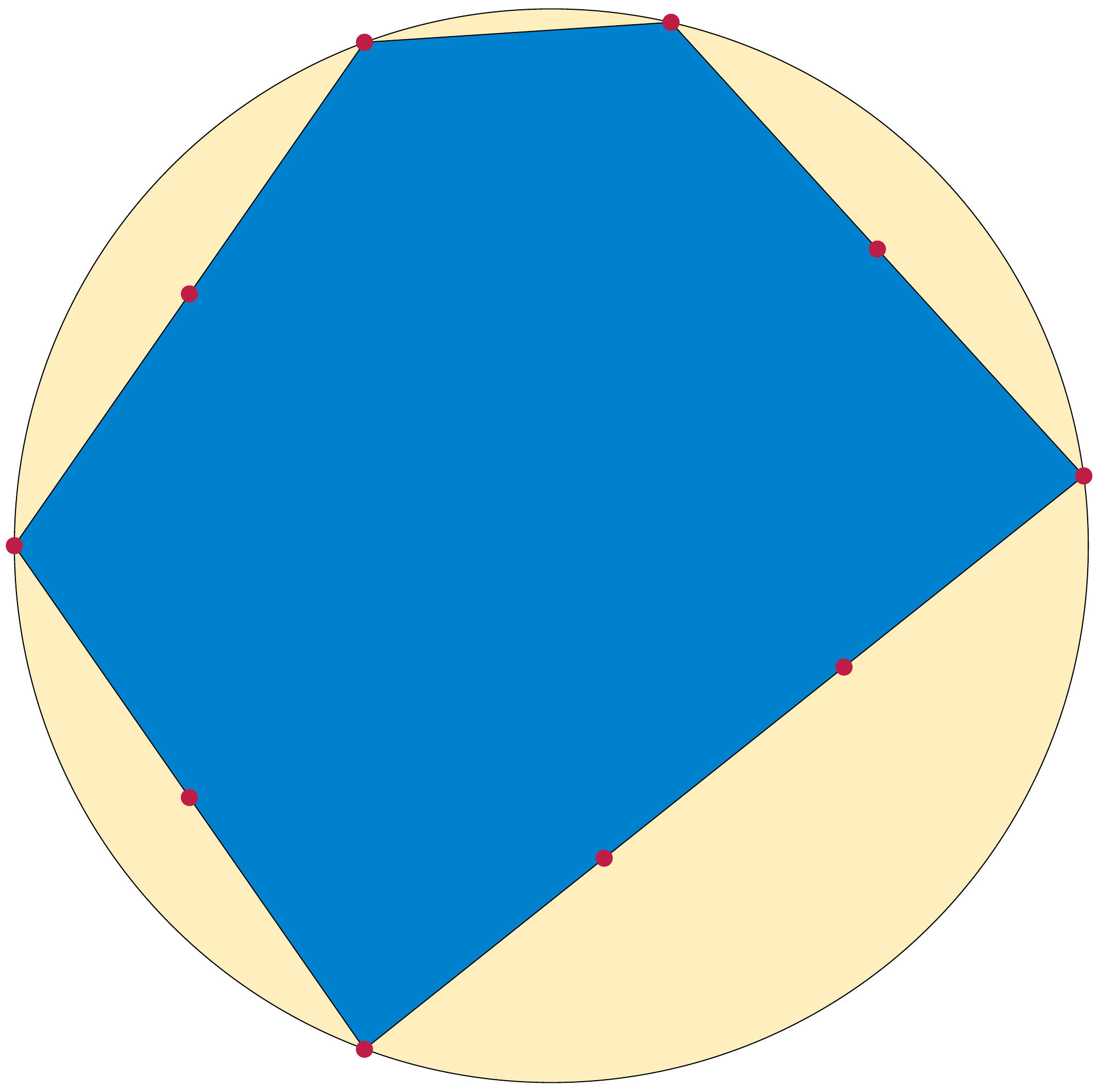}
\caption{A cyclic pentagon with side lengths $(1,2,3,2,2)$ and circumradius $\approx 1.75113$}
\label{fig:22233-cyclic-pentagon}
\end{figure}

\begin{lemma}
\label{lem:cyclic-existence}
Let $S=(s_1,s_2,\dots)$ be finitely many positive real numbers such that $\max S\le\sum S/(1+\pi/2)$, Then there exists a cyclic polygon with these edge lengths, containing the center of its circle.
\end{lemma}

\begin{proof}
Consider any circle $C_r$ of radius $r\ge\max S/2$, and inscribe a polygonal chain in $C_r$ whose edge lengths are drawn from $S$, choosing each vertex of the chain to be the next point clockwise around the circle at the correct distance. The condition that $r\ge\max S/2$ ensures that each segment of this chain will fit within $C_r$. For $r=\max S/2$, this chain will wrap around past its starting point, as each chord will span an arc of $C_r$ greater than its length, the chord of length $\max S$ will span the diameter, and the remaining chords will span an arc greater than a semicircle. However, for very large $r$, the chain will fail to wrap around $C_r$ even once. By the intermediate value theorem, there exists an $r$ for which the chain will wrap around exactly once, returning to its starting point.
\end{proof}

\cref{fig:22233-cyclic-pentagon} depicts an example, a cyclic pentagon with side lengths taken from the sequence $(1,2,3,2,2)$. Here $\max S=3$, $\sum S=10$, and $\sum S/(1+\pi/2)\approx 3.89$. For any cyclic polygon the side lengths can be permuted arbitrarily, corresponding geometrically to permuting wedges of the circumcircle; the circumradius remains unchanged. We note that a weaker condition, that $\max S\le\sum S/2$, ensures that a cyclic polygon exists (by a similar intermediate value argument with a counterclockwise turn adjacent to the longest segment) but it might not contain the circle center.

Cyclic pentagons are known to have high Galois complexity:

\begin{lemma}
\label{lem:s7-pentagon}
There exist cyclic pentagons with integer side lengths, containing the center of their circle, such that the Galois group of the minimal field containing their circumradius (over $\mathbb{Q}$) is $S_7$, the symmetric group on seven elements.
\end{lemma}

\begin{proof}
Varfolomeev~\cite{Var-MS-04} proves the conclusion, that the Galois group is $S_7$, but over the field 
\[\mathbb{Q}(\alpha_1,\alpha_2,\alpha_3,\alpha_4,\alpha_5)\]
where the $\alpha_i$ are five algebraically independent side lengths of the pentagon. By a standard argument applying the Hilbert irreducibility theorem (to a minimal polynomial of a primitive element of a splitting field) the same Galois group is obtained by choosing the tuple of $\alpha_i$ values from a Zariski-dense set of rationals, which (by density) can be chosen close enough to the tuple (1,1,1,1,1) to meet the conditions of \cref{lem:cyclic-existence}. Clearing denominators produces the desired integer lengths.
\end{proof}

\begin{theorem}
\label{thm:galois}
The coordinates of the geometric realizations of bipyramids with integer side lengths can neither be computed in the radical computation tree nor in the root computation tree models of Bannister et al.~\cite{BanDevEpp-JGAA-15}.
\end{theorem}

\begin{proof}
Uncomputability in the root computation tree follows for bipyramids over regular $p$-gons where $p$ is a large prime number for which the largest prime factor of $p-1$ is also large, as in other problems considered by Bannister et al.~\cite{BanDevEpp-JGAA-15}. Uncomputability in the radical computation tree model follows for bipyramids over the cyclic pentagons of \cref{lem:s7-pentagon}, because any closed form formulas in nested radicals for the coordinates of these pentagons could be combined to produce a formula for the pentagon's circumradius. Such a formula does not exist, because $S_7$ is not solvable.
\end{proof}

For related applications to origami of Varfolomeev's results, see Royo and Tramuns on the systems of numbers realizable by 3d origami~\cite{RoyTra-OSME-14}. Royo and Tramuns calculate explicitly that the Galois group for the cyclic pentagon with side lengths $(1,2,3,4,5)$ is $S_7$. Instead of considering bipyramids and using Alexandrov's theorem, they follow a model of 3d origami folding in which it is possible to flatten the base of an open pyramid against a tabletop.

\section{Space and counting complexity}

Although it is difficult to find a flat-folded state for a crease pattern (\cref{thm:eth}), it is common when designing crease patterns to have a flat-folded state already in mind. Instead one must find folding steps that bring the crease pattern to its known flat-folded state. This is a \emph{reconfiguration problem}, in which one must navigate a combinatorial structure of states (partial foldings) and moves between pairs of states (individual folding steps) to reach a goal (the eventual flat folding).

Past work on reconfiguration in origami includes the work of Akitaya et al. on face flips in origami tessellations~\cite{AkiDujEpp-JoCG-20}. Here, an instance is defined by an unlabeled repeating crease pattern, such as a square grid, triangular grid, or the parallelograms of the Miura-ori. The states are \emph{locally foldable} mountain--valley assignments to the creases: assignments such that a small neighborhood of each vertex can be folded flat. The moves between states, in this work, reverse all creases surrounding a single face of the crease pattern, when this reversal leads to another locally foldable state. As Akitaya et al. showed, for a tessellation of $n$ equilateral triangles, any two locally foldable states can be connected to each other by $O(n)$ moves, but it is $\mathsf{NP}$-hard to find the shortest move sequence.

Another related result in origami reconfiguration comes from a second paper by Akitaya et al.~\cite{AkiDemHor-JoCG-20}. They showed that, for rigid origami, it is weakly $\mathsf{NP}$-complete to determine whether a given crease pattern has a continuous motion from its initial flattened state that simultaneously flexes all creases. In the same model, it is strongly $\mathsf{NP}$-complete to determine whether a subset of creases can be flexed. Here, weak $\mathsf{NP}$-completeness means that the problem is hard when the crease coordinates are written in the usual way as binary numbers, while strong $\mathsf{NP}$-completeness means that it remains hard even when these numbers are written in unary. In some sense, the difficulty of a strongly $\mathsf{NP}$-complete problem is purely combinatorial, whereas a weakly $\mathsf{NP}$-complete problem may require very large and precisely-determined numerical coordinates for its hard instances.

However, for reconfiguration problems, space complexity is more natural than nondeterministic time complexity: many reconfiguration problems are $\mathsf{PSPACE}$-complete, a stronger result than $\mathsf{NP}$-completeness. In this section, we introduce a simple and natural toy model of reconfiguration in origami flat folding for which we prove reconfiguration to be $\mathsf{PSPACE}$-complete. Our proof, by a reduction from \emph{nondeterministic constraint logic}, can be adapted to also prove that counting the number of states in the same toy problem is $\mathsf{\#P}$-complete. 

\subsection{Flaps and flips}

\begin{figure}[t]
\centering\includegraphics[width=0.35\textwidth]{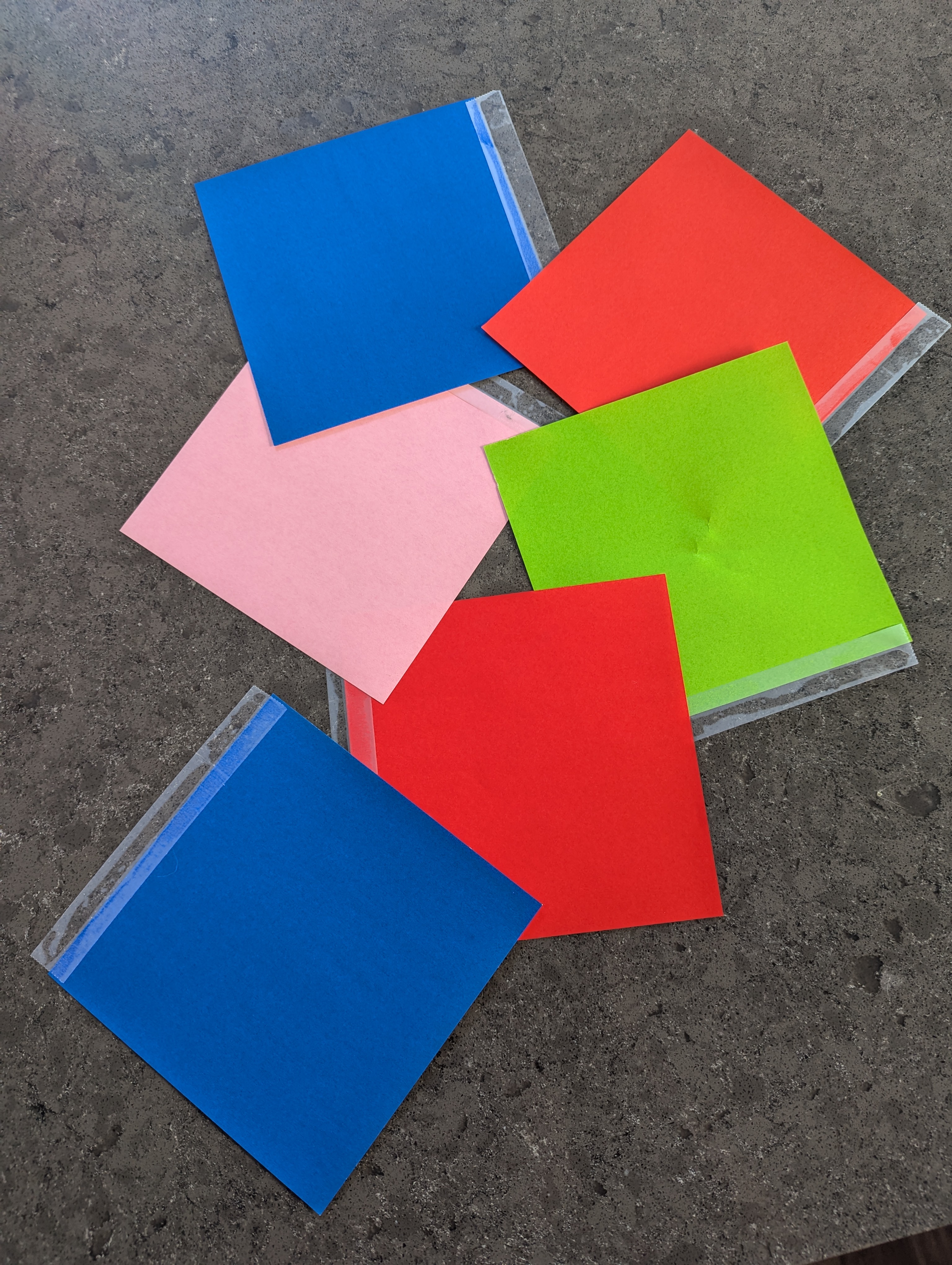}
\caption{An instance of flaps and flips, consisting of five squares of origami paper, each taped by one edge to a tabletop.}
\label{fig:fnf}
\end{figure}

We define here a family of state spaces for origami reconfiguration problems, which we call \emph{flaps and flips}. An instance of a flaps and flips problem consists of equal-sized squares of origami paper (\emph{flaps}), each attached to the surface of a rigid and flat tabletop by a hinge along the full length of one edge. Hinges are not allowed to cross through other flaps or hinges. A \emph{state} of this instance is, essentially, a flat folding in the sense described in \cref{sec:flat}: each flap must lie flat on the table, with its hinged edge in its fixed position on the table. Each two flaps that overlap must have a consistent above--below relation within their region of overlap, although the above--below relations among three or more flaps may have cycles. When one flap lies below another, the flap that is below cannot overlap the hinge of the flap that is above. The moves of this problem (\emph{flips}) change the position of a single flap to any other consistent position. We do not require rigid motion: a flap can be pulled out of a pocket formed by other flaps, or inserted into another pocket.

For an example, see \cref{fig:fnf}. The lower blue flap can flip to the other side of its hinge, as no other flap gets in its way. The central three flaps have cyclic above--below relations. The green flap can also flip, even though it is partially covered by a red flap, as the resulting state remains valid. No other flap can be moved to a different consistent position, until the flaps blocking them move.

As is standard for reconfiguration we define two versions of the question of whether a flaps and flips problem can be reconfigured. We also define an associated counting problem:
\begin{itemize}
\item Pairwise connectivity: Given two valid states of the same flaps and flips instance, does there exist a sequence of flips that transforms one of the two states into the other?
\item Global connectivity: Is every pair of valid states of a given flaps and flips instance connected by a sequence of flips?
\item Counting states: Given a flaps and flips instance, how many valid states does it have?
\end{itemize}

\subsection{Conversion from constraint logic}

We prove the $\mathsf{PSPACE}$-completeness of flaps and flips reconfiguration by a reduction from \emph{nondeterministic constraint logic} (NCL), a reconfiguration problem frequently used to prove hardness for puzzles and games~\cite{HeaDem-TCS-05,HeaDem-09}  Instances of nondeterministic constraint logic take the form of 3-regular directed graphs with edges colored blue and red, and with one or three blue edges per vertex. Each vertex must have at least one blue or two red edges directed into it. The allowed moves are to reverse one edge at a time, maintaining these constraints. \cref{fig:ncl-example} depicts an example in which the three edges incident to the yellow vertex can be reversed; the other edges cannot change orientation without violating the constraints. Reaching one orientation of the edges from another, or testing whether all orientations are connected, is $\mathsf{PSPACE}$-complete, even for planar graphs of bounded bandwidth~\cite{vdZ-IPEC-15}.

\begin{figure}[t]
\centering\includegraphics[width=0.3\textwidth]{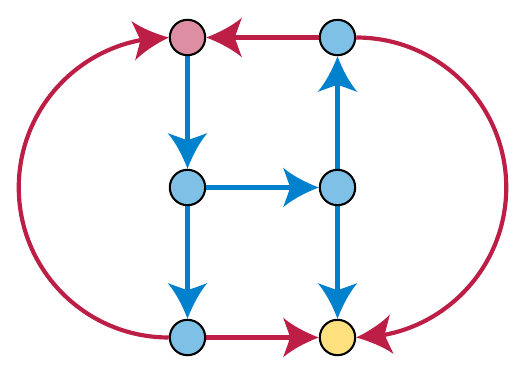}
\caption{An NCL instance. Each light blue vertex has an incoming blue edge, and the light red vertex has two incoming red edges. At the yellow vertex, both conditions are true, allowing any edge incident to this vertex to be reversed.}
\label{fig:ncl-example}
\end{figure}

\begin{figure}[t]
\centering\includegraphics[scale=0.2]{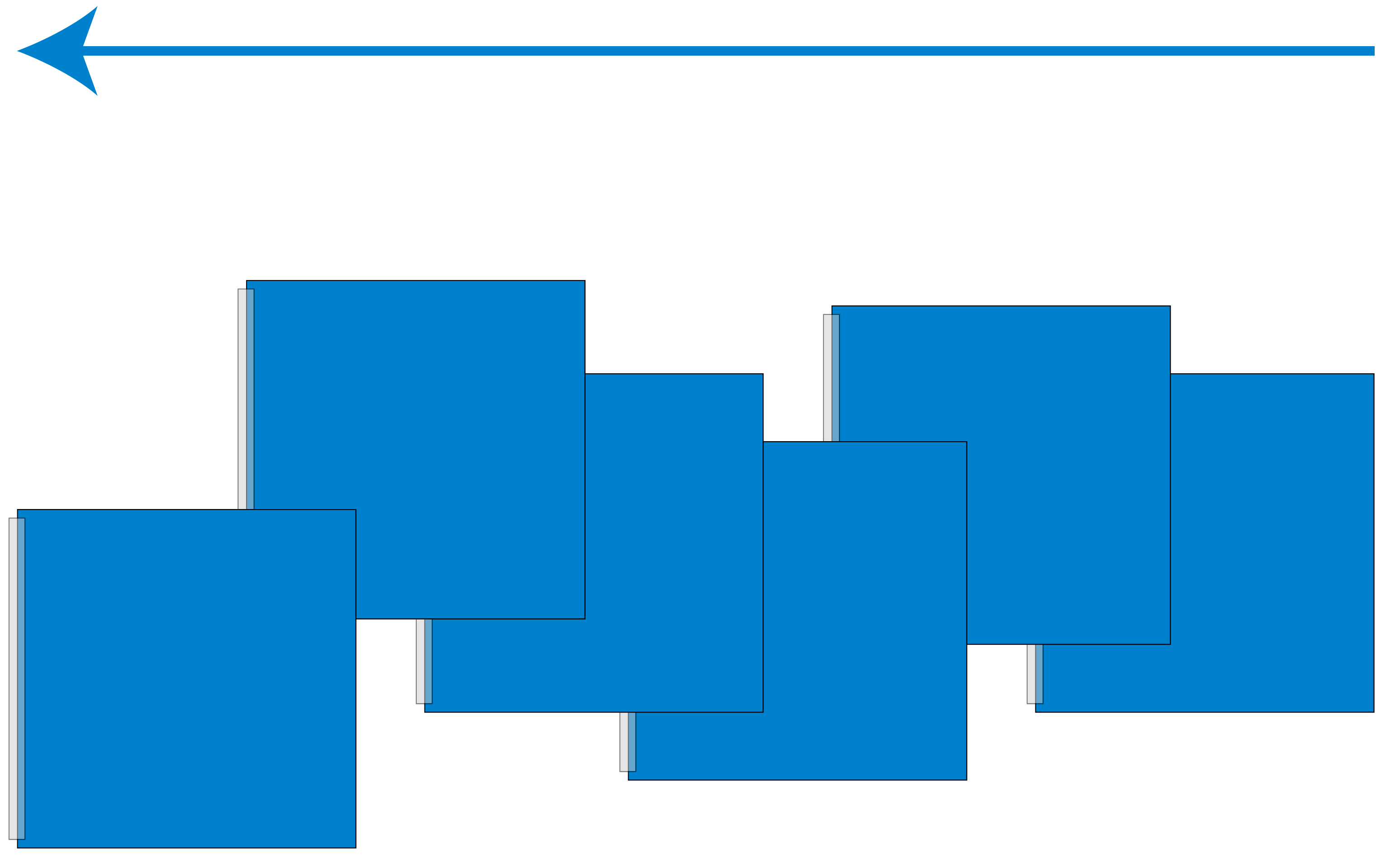}
\caption{Gadget for converting a directed edge of a nondeterministic constraint logic into flaps and flips.
Each flap, when flipped to one side of its hinge, will overlap the hinge of the next flap in that direction. The corresponding edge points in the direction of the hinges.}
\label{fig:arrow-nt}
\end{figure}

To prove $\mathsf{PSPACE}$-hardness for flaps and flips, we represent NCL edges and vertices using gadgets. For an edge gadget, lay out a sequence of flaps, each of which (when flipped to either side) covers part of the hinge of a neighboring flap in the sequence, as shown in \cref{fig:arrow-nt}. The arrowhead of the corresponding directed edge points to the uncovered hinge at one end of this sequence. To reverse the edge, flip its flaps in order, from the arrowhead to the tail. The intermediate states of these flips will leave the edge having two tails and no arrowheads, but this inconsistent state will not free up any other flaps at the ends of the arrow, so this possibility will not be problematic for reconfiguration. (We will need to take this possibility into account when we count states.) Two flaps in this gadget cannot point away from each other (with hinges on the far edges of the flaps from each other), because they would overlap, and whichever flap is below the other would make a disallowed crossing through the hinge of the other flap, so these are the only possible states of this gadget.

\begin{figure}[t]
\centering\includegraphics[scale=0.2]{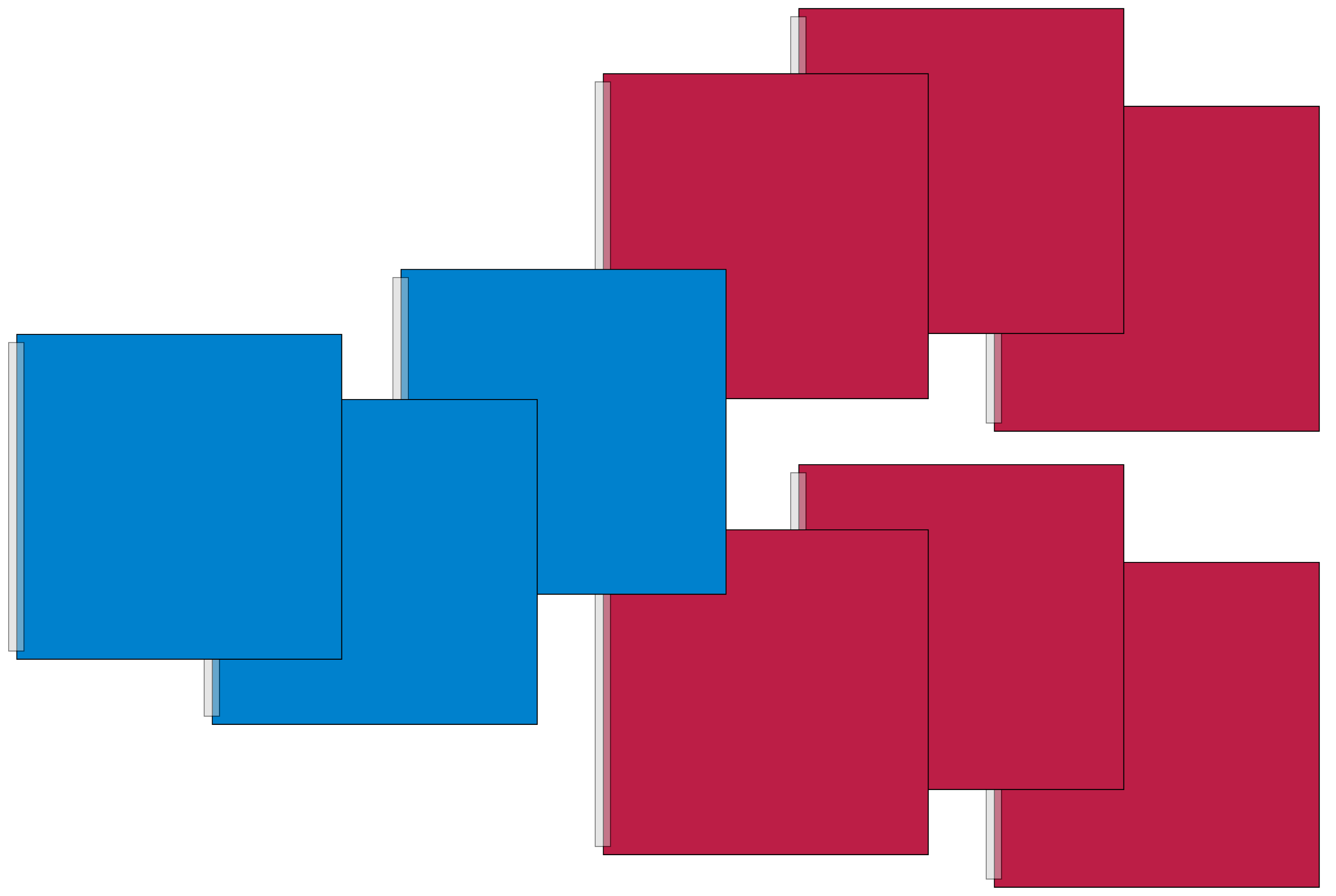}
\caption{Gadget for converting a blue--red--red NCL vertex into flaps and flips.
As shown, both incident red edges point into the vertex (with the hinges of each flap pointing towards the gadget), and the blue flap is free to flip. The red flaps can flip only when the blue edge points into the vertex (with the hinge of the blue flap pointing towards the gadget).}
\label{fig:and-gate}
\end{figure}

For a blue--red--red vertex, arrange the flaps at the ends of the three edges so that the blue flap can overlap both red hinges, but each red flap can only overlap the blue hinge (\cref{fig:and-gate}). If the blue flap hinge is covered, the blue flap itself must be flipped away from the vertex. In terms of the underlying NCL instance, the blue edge points in, as required. If the blue flap is not flipped away, it covers the two red flaps, and the two red edges point in. So this arrangement of flaps has the same behavior as a blue--red--red vertex is required to have: its allowed states are the ones where the blue edge points in, the two red edges both point in, or both of those things happen.

\begin{figure}[t]
\centering\includegraphics[scale=0.2]{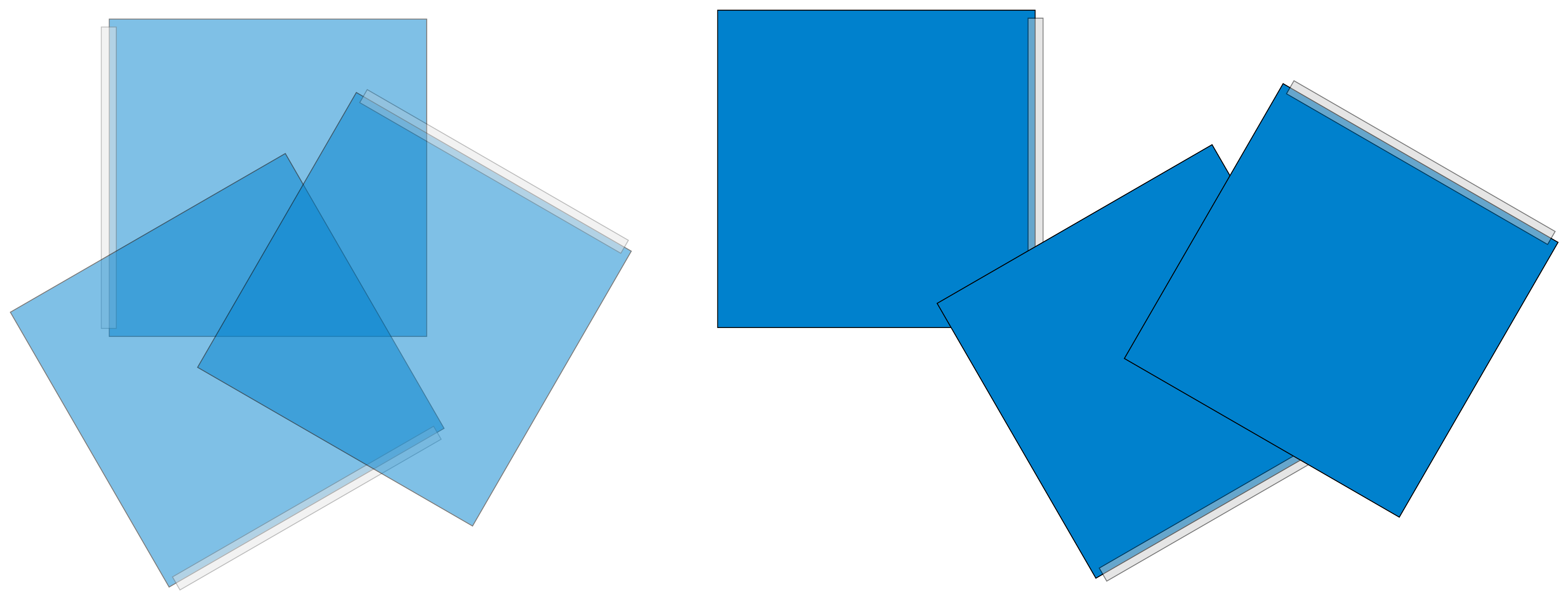}
\caption{Gadget for converting a blue--blue-blue NCL vertex into flaps and flips. The left side of the figure shows a disallowed state, with all three edges pointing out; this is impossible, because it would produce a cyclic above--below ordering for the three flaps in the region where all three overlap. The right side shows a state with one edge pointing in, and the other two flaps free to flip.}
\label{fig:or-gate}
\end{figure}

The blue-blue-blue vertex is shown in two arrangements in \cref{fig:or-gate}, but only one can be flat-folded. In the left part of the illustration, three flaps at the ends of three blue edges flip towards the vertex, corresponding to the forbidden state where three edges are directed out from the vertex. Each flap overlaps the hinge of another flap, forming cyclic above-below relations that cannot be realized in the central area where all three flaps overlap. However, if at least one flap is flipped away from the vertex (corresponding to the arrow pointing in), it is realizable, as shown on the right. In this state, the other two arrows can flip freely, unobstructed by the flipped-out flap or each other.

\begin{figure}[t]
\centering\includegraphics[scale=0.2]{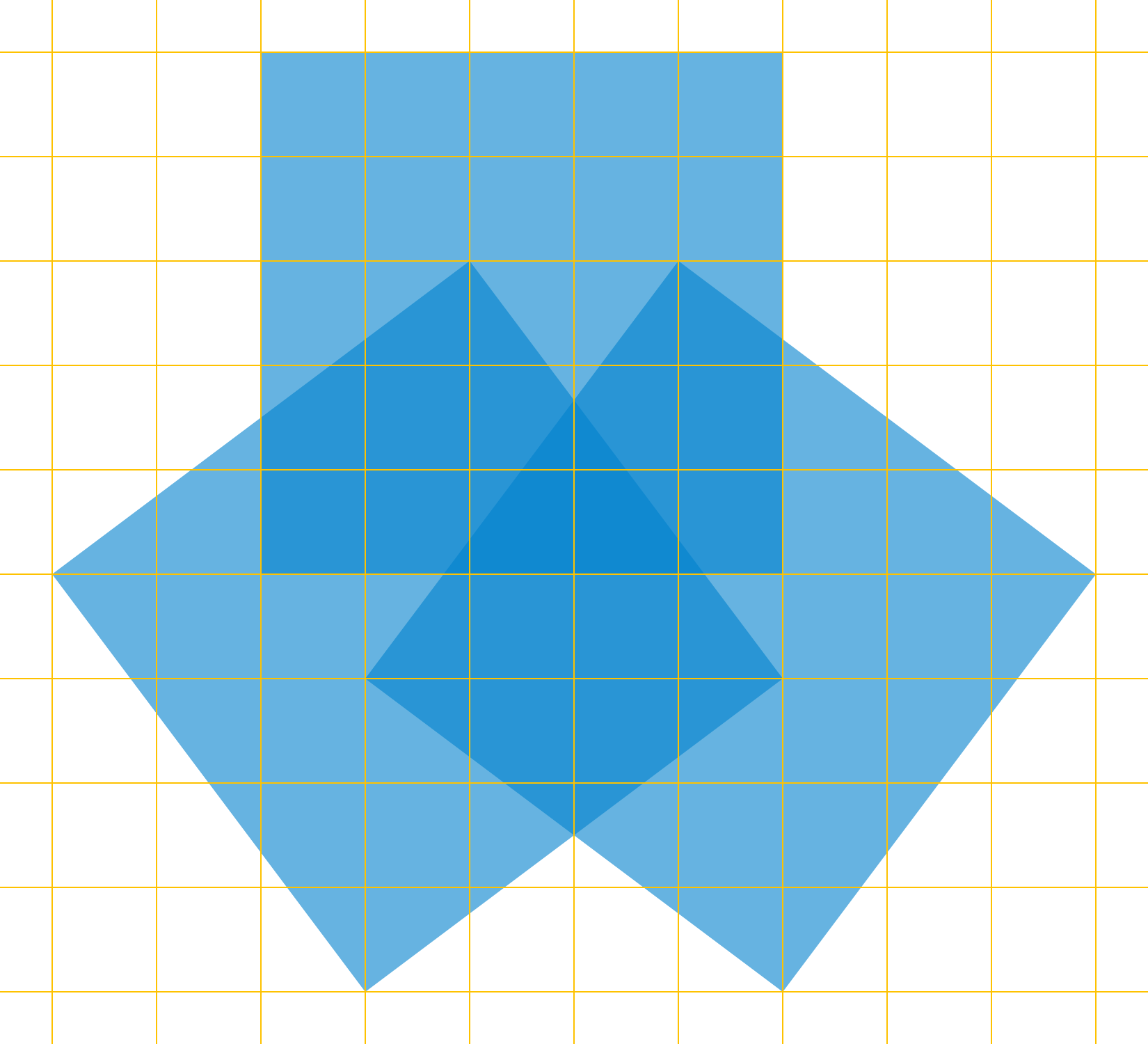}
\caption{Squares with integer vertex coordinates and side length five are sufficient to realize a blue--blue-blue vertex gadget.}
\label{fig:grid-or}
\end{figure}

To get a hard instance for flaps and flips, translate the vertices and edges of a hard planar NCL instance into gadgets, with the three central flaps of each vertex gadget forming the end flaps of an edge gadget. Cubic planar graphs can be drawn with vertices on a square grid of linear dimensions and with edges following paths of grid edges~\cite{Val-TC-81,DidLioOrt-SODA-20}. From such a layout, choose a flap size small enough relative to the grid size so that each vertex gadget can fit in a small neighborhood of a grid point, without interfering with other features of the layout, so that each edge gadget can be routed around the corners of a grid path by a sequence of obtuse-angle turns, and so the edge gadgets can be connected to the vertex gadgets appropriately. To avoid issues of coordinate representation such as those raised in \cref{sec:galois}, we use integer coordinates for the vertices of each flap, with squares of side length $5$, allowing both axis-aligned flaps and flaps with sides whose slopes are $\pm 3/4$ and $\pm 4/3$. These slopes are sufficient to turn the flaps in an edge gadget around corners in a grid layout, and to realize the three overlapping triangles of a blue-blue-blue vertex gadget (\cref{fig:grid-or}).

Define a \emph{canonical state} of a flaps and flips problem derived in this way from NCL to be a state in which each edge gadget has a consistent orientation, with an arrowhead at one end instead of two tails. Let us now formalize the correspondence between these problems:

\begin{lemma}
\label{lem:ff-redux-correctness}
For a flaps and flips problem derived from NCL, every NCL state corresponds to a canonical flaps and flips state, and every two adjacent NCL states correspond to two canonical flaps and flips states that are connected by flips. Every flaps and flips state corresponds to a canonical state that can be reached by flipping only flaps of edge gadgets that are not already in a canonical state. Every two adjacent flaps and flips states correspond to canonical states that represent either the same NCL state or adjacent NCL states.
\end{lemma}

\begin{proof}
The correspondence from NCL states to canonical flaps and flips states is defined gadget-by-gadget, according to the description of the edge gadgets above. Our vertex gadgets are defined in such a way that, for a valid NCL state, the corresponding canonical flaps and flips state will have a consistent arrangement of flaps within each vertex gadget. An NCL edge reversal corresponds to flips of all flaps in an edge gadget, ordered from the arrowhead of the edge gadget to its tail.

To map arbitrary flaps and flips states to canonical states, choose arbitrarily a distinguished endpoint of each edge gadget, and replace each edge gadget that is not already in a canonical position (with an arrowhead at one end) by a different state of the same gadget that has the arrowhead at the distinguished endpoint. Each state can reach its corresponding canonical state by flipping the flaps of its non-canonical edge gadgets, ordered from the point at which two consecutive interior flaps of the gadget are oriented non-canonically towards the distinguished endpoint.

Two adjacent flaps and flips states differ in the position of one flap and have canonical states that can differ only in the orientation of the edge gadget containing that flap. The corresponding NCL states differ only in the orientation of this single edge, or not at all.
\end{proof}

\begin{theorem}
\label{thm:pspace}
Pairwise connectivity and global connectivity for flaps and flips, for flaps with integer coordinates, is $\mathsf{PSPACE}$-complete, even for instances with bounded treewidth as defined in \cref{sec:treewidth}.
\end{theorem}

\begin{proof}
Pairwise connectivity problem is in $\mathsf{PSPACE}$, because a nondeterministic $\mathsf{PSPACE}$ algorithm can search for flips connecting one state to the other, and nondeterministic and deterministic $\mathsf{PSPACE}$ are the same by Savitch's theorem.
Global connectivity can be solved (in complementary form) in $\mathsf{PSPACE}$ by nondeterministically guessing two disconnected states and applying the deterministic pairwise connectivity algorithm to verify disconnectedness.
The conversion from planar NCL to flaps and flips, for either problem, is a polynomial time many-one reduction, whose correctness follows from \cref{lem:ff-redux-correctness}. The theorem follows from the known $\mathsf{PSPACE}$-completeness of pairwise connectivity and global connectivity for planar NCL instances of bounded bandwidth~\cite{vdZ-IPEC-15}, and the fact that bounded bandwidth implies bounded treewidth.
\end{proof}

\subsection{Counting}
\label{sec:counting}

Combinatorial counting problems are modeled in computational complexity by the class $\mathsf{\#P}$, whose members are computational problems that output the number of accepting paths of a nondeterministic polynomial-time Turing machine. These include counting the flat-folded states of a finite crease pattern or flaps and flips instance, as one can nondeterministically generate the above--below orderings for each cell and accept a choice when it forms a valid flat folding. A problem is $\mathsf{\#P}$-complete if it belongs to $\mathsf{\#P}$ and is universal in the sense that every $\mathsf{\#P}$ problem can be reduced to it. However, there is a choice in what kind of reductions to allow:
\begin{itemize}
\item A \emph{Turing reduction} from problem $X$ to problem $Y$ is an algorithm for solving $X$ using a subroutine for $Y$, in polynomial time (not counting time within subroutine calls).
\item A \emph{polynomial-time counting reduction} from problem $X$ to problem $Y$ is a Turing reduction that calls its subroutine exactly once. Alternatively, it can be described in terms of two polynomial-time algorithms, to transform instances of problem $X$ into instances of problem $Y$, and to transform the subroutine output into the output for~$X$.
\item A \emph{parsimonious counting reduction} from problem $X$ to problem $Y$ calls the subroutine for $Y$ only once and then immediately returns the same number. Equivalently, it transforms instances of $X$ into instances of $Y$ that have the same output.
\end{itemize}
These three reduction types are successively more restrictive, and each is closed under compositions. We will follow a chain of reductions from the problem of counting matchings in arbitrary bipartite graphs, to matchings in 3-regular bipartite graphs, to NCL states, to flaps and flips states. Counting bipartite matchings was proven $\mathsf{\#P}$-complete by Valiant in the work that introduced $\mathsf{\#P}$-completeness, using a Turing reduction~\cite{Val-TCS-79}. Ben-Dor and Halevi replaced this by a polynomial-time counting reduction~\cite{BenHal-ISTCS-93}. A parsimonious reduction from bipartite matching to 3-regular bipartite matching is provided by Dagum and Luby, who needed parsimony for related results in the hardness of approximation~\cite{DagLub-TCS-92}. We provide new polynomial-time counting reductions, from 3-regular bipartite matching to NCL, and from the resulting special cases of NCL to flaps and flips, proving that counting flaps and flips states is $\mathsf{\#P}$-complete under polynomial-time counting reductions. We will need at least counting reductions, as NCL is not $\mathsf{\#P}$-complete under parsimonious reductions.

Our reduction from NCL to flaps and flips is based on the many-one reduction that we used in \cref{thm:pspace} to prove $\mathsf{PSPACE}$-completeness. However, we need to handle nonplanar instances, because planar bipartite matchings can be counted in polynomial time~\cite{Kas-GTTP-67}. Nonplanarities could be removed in the transition from matching to NCL, using an appropriate crossover gadget, but we instead supply a parsimonious crossover gadget for flaps and flips.

\begin{figure}[t]
\centering\includegraphics[width=0.35\textwidth]{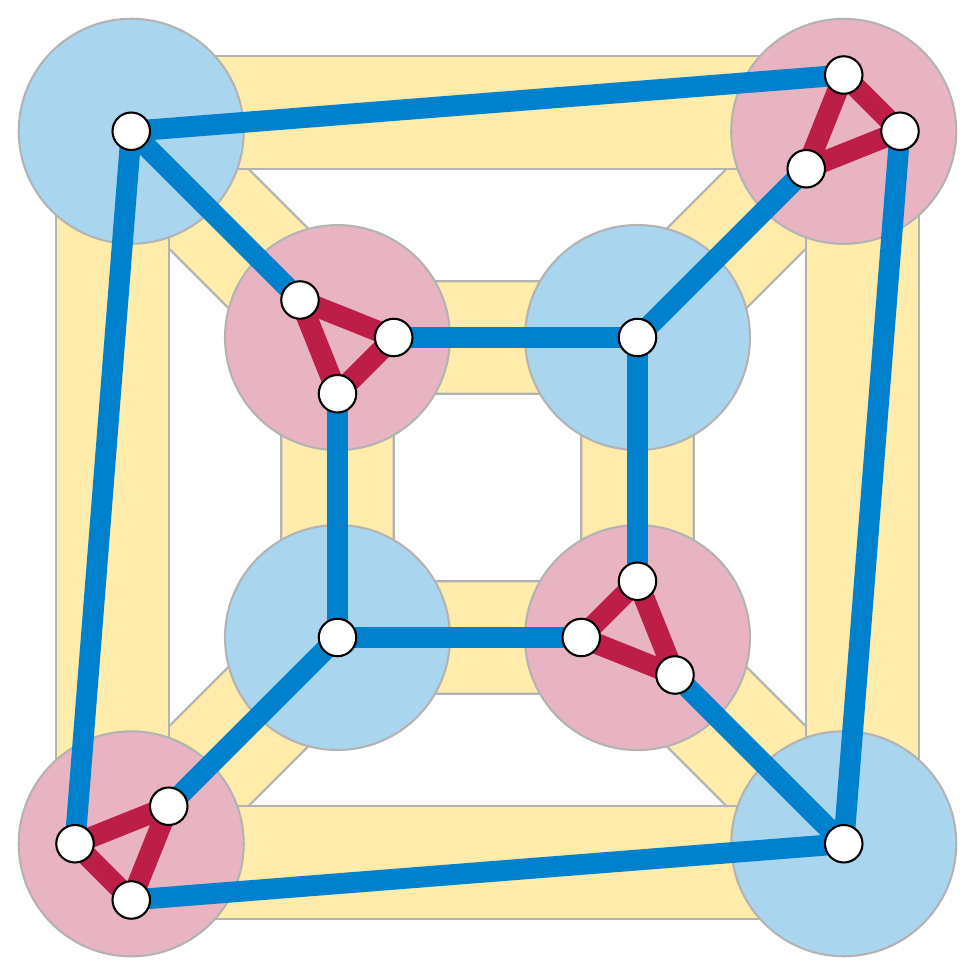}
\caption{Reduction from counting perfect matchings in a 3-regular bipartite graph (here the graph of a cube, with vertices drawn as the large light-shaded circles and thick yellow edges) to an instance of nondeterministic constraint logic (small white vertices and thin red and blue edges).}
\label{fig:cube-to-cgs}
\end{figure}

\begin{lemma}
\label{lem:match2ncl}
There exists a polynomial-time counting reduction from counting perfect matchings in a 3-regular bipartite graph to the problem of counting NCL states.
\end{lemma}

\begin{proof}
The reduction is depicted in \cref{fig:cube-to-cgs}: replace the vertices on one side of the bipartition of the 3-regular graph with triangles of red edges, each adjacent to one of the three neighbors of the replaced vertex, and color all remaining edges blue. If there are $n$ red (replaced) vertices and $n$ blue vertices in the given graph, the resulting constraint logic instance has $4n$ vertices, $3n$ blue edges, and $n$ red triangles. Each triangle can supply at most one pair of incoming red edges to one triangle vertex. The $3n$ blue edges and $n$ triangles are exactly enough to satisfy the $4n$ NCL vertices, so each vertex is satisfied in exactly one way. The blue edges that satisfy a blue vertex form a perfect matching, their tails can only be satisfied by two red edges, and the remaining blue edges must be directed away from their blue endpoints. This leaves $n$ red edges that are free to be oriented in either direction. In this way, every perfect matching of the given 3-regular graph corresponds to exactly $2^n$ NCL solutions. We can count perfect matchings by counting NCL solutions and dividing by $2^n$.
\end{proof}

\begin{corollary}
\label{cor:ncl-count}
Counting NCL solutions is $\mathsf{\#P}$-complete under polynomial-time counting reductions.
\end{corollary}

\begin{observation}
\label{obs:ncl-nonunique}
Every solvable NCL instance has more than one solution.
\end{observation}

\begin{proof}
Assign weight 2 to blue edges and weight 1 to red edges, and consider the following cases:
\begin{itemize}
\item The edges at any vertex have total weight 4 (red--red--blue) or 6 (blue--blue--blue). So if there is a blue--blue--blue vertex, the average weight per vertex is more than 4. The average weight of incoming edges per vertex in a valid solution is half of that quantity, more than two, and therefore some vertex has incoming weight three or more. At such a vertex, one of the edges must be free to be reversed, producing another satisfying orientation.
\item Otherwise, there are only red--red--blue vertices. The average weight per vertex is exactly 4, the average incoming weight per vertex is exactly 2, and to satisfy each vertex the incoming weight must be exactly 2. The outgoing weight at each vertex is also exactly 2, so reversing every edge simultaneously produces another satisfying orientation.\qedhere
\end{itemize}
\end{proof}

\begin{corollary}
Counting NCL solutions is not $\mathsf{\#P}$-complete under parsimonious reductions.
\end{corollary}

\begin{proof}
Instances in $\mathsf{\#P}$ with one solution cannot be parsimoniously reduced to NCL.
\end{proof}

\begin{figure}[t]
\centering\includegraphics[scale=0.2]{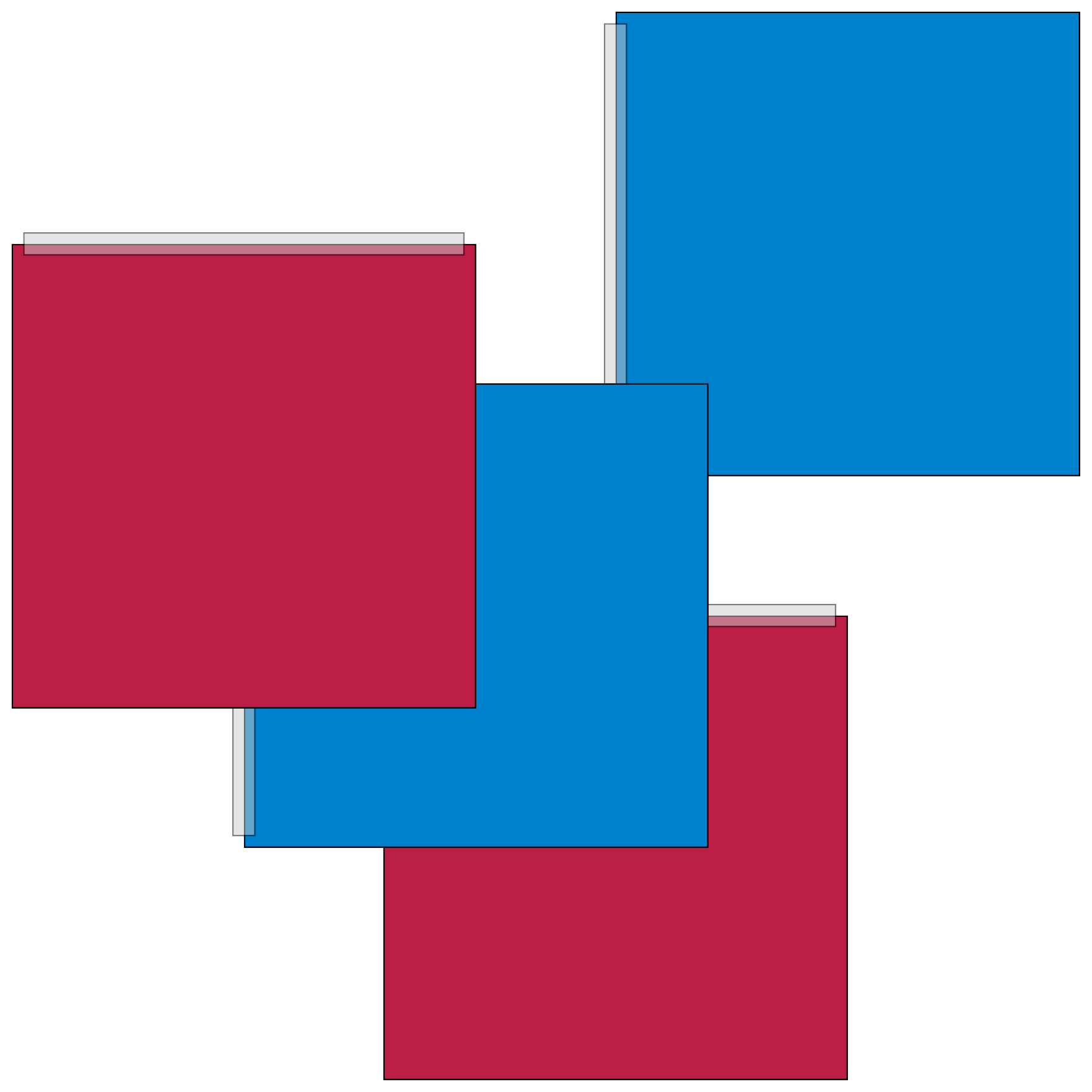}
\caption{A crossover gadget for flaps and flips. The two edge gadgets that cross in this gadget are unconstrained by each other.}
\label{fig:ff-crossover}
\end{figure}

\begin{theorem}
\label{thm:count}
Counting flaps and flips states is $\mathsf{\#P}$-complete under counting reductions.
\end{theorem}

\begin{proof}
We provide a polynomial-time counting reduction to flaps and flips from the NCL instances created by \cref{lem:match2ncl}, in which the red edges form disjoint triangles connected in pairs by blue edges. To do so, we use the same gadgets as for \cref{thm:pspace}, applied to a grid drawing of the NCL instance in which the red triangles are drawn with exactly two grid edges for each red edge of the constraint logic instance, and in which all crossings involve two blue edges and occur at a grid vertex. At each crossing, we use the crossover gadget depicted in \cref{fig:ff-crossover}, for which the two edge gadgets that cross do not interact with each other: any consistent state of both edge gadgets, taken separately, corresponds to a unique consistent state of the crossover gadget.

Because we arrange this drawing with the same length for each red edge, we can replace each red edge by an edge gadget with the same number of flaps, $k$.
As argued in \cref{lem:match2ncl}, in the consistent states of the NCL instance, most edges are fixed in orientation, required to uniquely satisfy some vertex, and the corresponding edge gadgets are likewise fixed in a single flat-folded state. The only parts of the flaps and flips instance that can flip are in the edge gadgets corresponding to the reversable red edges. If there are $n$ red triangles, there are $n$ reversible red edges and $n$ flippable red edge gadgets. Each flippable red edge gadget has $k+1$ flat-folded states, obtained by partitioning its flaps into two contiguous subsets and flipping each flap away from the ones in the other subset. Therefore, for each connected component of $2^n$ NCL states, there will be a corresponding connected component of exactly $(k+1)^n$ flaps and flips states. The number of NCL states can be calculated from the number of flaps and flips states by multiplying by $2^n/(k+1)^n$.

The chain of polynomial-time counting reductions from arbitrary problems in $\mathsf{\#P}$ to bipartite matchings, 3-regular bipartite matchings, NCL, and flaps and flips, proves that counting flat-folded states of flaps and flips instances is $\mathsf{\#P}$-complete, as stated.
\end{proof}

\begin{figure}[t]
\centering\includegraphics[scale=0.2]{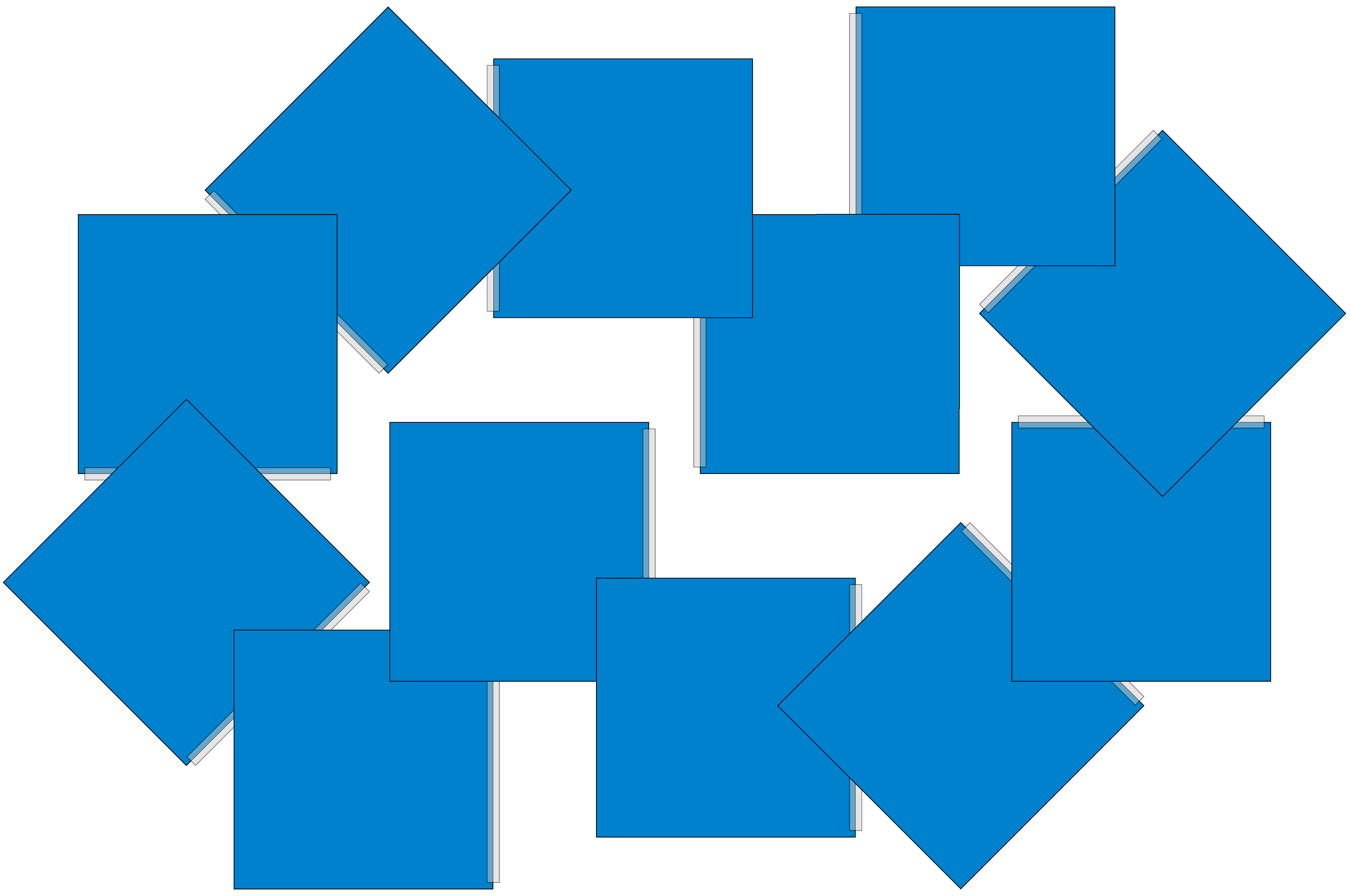}
\caption{Flaps and flips instance with a unique flat-folded state}
\label{fig:ff-unique}
\end{figure}

The same issue that prevents parsimonious completeness for nondeterministic constraint logic, the inability to find an instance with a unique solution (\cref{obs:ncl-nonunique}) does not apply to flaps and flips. There exist flaps and flips instances with unique flat-folded states (\cref{fig:ff-unique}). We do not know whether counting flaps and flips states is $\mathsf{\#P}$-complete under parsimonous reductions.

The fixed-parameter tractable algorithm of \cref{thm:fpt}, and the ETH-based lower bound on its complexity in \cref{thm:eth}, apply equally well to counting the flat-folded states of a crease pattern, parameterized by its ply and treewidth. For flaps and flips, the definitions of ply and treewidth do not apply directly, as they assume a unique local flat folding, which in flaps and flips varies according to which side of its hinge each flap lies on. Instead, we may define the ply and treewidth for a local flat folding that includes two copies of each flap, one on each side of its hinge. With this modified definition, counting flaps and flips states is fixed-parameter tractable in ply and treewidth, as in \cref{thm:fpt}. A reduction from 3SAT to flaps and flips via NCL, with the same overall form as the $\mathsf{NP}$-completeness proof for flat folding of Bern and Hayes~\cite{BerHay-SODA-96} and Akitaya et al.~\cite{AkiCheDem-JCGCGG-15} (with horizontal edge gadgets representing variables, vertical edge gadgets connected to them by red--red--blue vertex gadgets and crossing the other variables using crossover gadgets, and blue--blue--blue vertex gadgets representing clauses) shows that, under the exponential time hypothesis, the exponential dependence on treewidth of this algorithm cannot be improved, as in \cref{thm:eth}.

\section{Infinite crease patterns}

Recently Hull and Zakharevich~\cite{HullZak-23} and Assis~\cite{Ass-24} showed that, in a certain sense, folding crease patterns with infinitely many creases, on infinite sheets of paper, can simulate universal computation, far beyond the complexity classes discussed above. After definitions in \cref{sec:fundamentals}, we review this work in \cref{sec:review} before strengthening it to show the undecidability of flat foldability for infinite crease patterns on finite sheets of paper.

\subsection{Fundamentals}
\label{sec:fundamentals}

Crease patterns with infinitely many creases do not require infinitely large sheets of paper. Indeed, the artistic work of Tomoko Fuse includes fractal crease patterns on finite paper that, taken to their limit, would include infinitely many folds~\cite{Fus-IF-20}. Such infinite crease patterns raise mathematical difficulties beyond the finite case. For instance, a finite crease pattern has a local flat folding that is essentially unique (\cref{obs:local}), but for infinite crease patterns this may be untrue. For example, the infinite crease pattern of \cref{fig:bowtie} folds flat to two triangles that meet at a point. Changing the angle at which the triangles meet would produce a different but equally consistent locally flat folding. Additionally, the above--below orderings that define flat foldings are discrete rather than continuous (unlike local flat foldings, which are continuous functions). Therefore, we cannot use continuity and limits to determine whether these orderings are consistent at accumulation points of creases.

To avoid such definitional issues, we will consider crease patterns with the following property. Define a point of a crease pattern to be \emph{crumpled} if every neighborhood of the point intersects infinitely many creases. Define a crease pattern to be \emph{tame} if it is defined on a simply-connected open set (such as the plane, open half-plane, or an open square or disk), has no crumpled points, and every two points are connected by a curve of finite length within the crease pattern. These requirements are not met for the crease pattern of \cref{fig:bowtie}. However, crease patterns based on periodic tessellations of the plane, such as that of Hull and Zakharevich~\cite{HullZak-23}, are tame, and a tame crease pattern can also take the form of a fractal system of creases on an open square with crumpled points only at its boundary (where they are not considered to be part of the crease pattern). For a tame crease pattern, every curve of finite length between two points crosses only finitely many creases, because otherwise its crossings with creases would have a limit point. Tame crease patterns have the following analogue of \cref{obs:local}:

\begin{figure}[t]
\centering\includegraphics[width=0.6\textwidth]{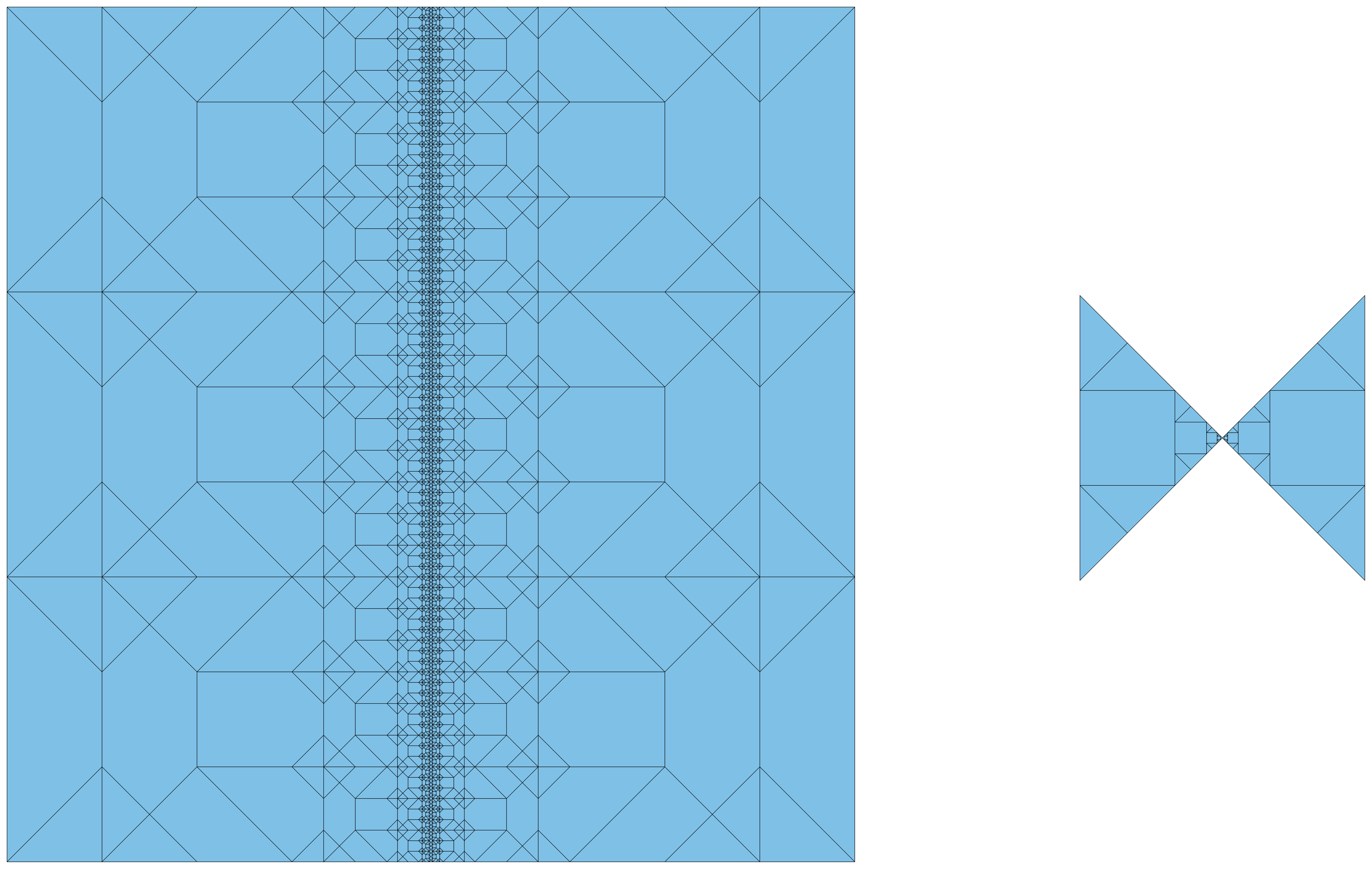}
\caption{How to fold a bowtie}
\label{fig:bowtie}
\end{figure}

\begin{observation}
\label{obs:tame}
A local flat folding of a tame crease pattern is unique up to rigid motion.
\end{observation}

\begin{proof}
Let $\varphi$ be a local flat folding of a tame crease pattern $P$. Choose an arbitrary crease $c_0$, and an open neighborhood of $c_0$ that intersects no other crease, and map this neighborhood into the plane by a restriction of $\varphi$. Then, for each other crease $c_i$, we can find the restriction of $\varphi$ to a neighborhood of $c_i$ by following a  curve of finite length from $c_0$ to $c_i$ and composing finitely many reflections for each crease intersected by this curve. Each reflection across a crease gives $\varphi$ on the other side of the crease, so this composition of reflections remains equal to $\varphi$, regardless of the choice of curve. Because $P$ is tame, every point that is not on a crease has a closest crease, and is in its open neighborhood, so this determines $\varphi$ at all points. Thus, $\varphi$ is completely fixed given the initial mapping of the neighborhood of $c_0$, which is (in that neighborhood) unique up to rigid motions.
 \end{proof}
 
For any finite set $S$ of creases of a crease pattern, define the \emph{nearby region} of $S$ to be the set of points that can be connected to $S$ by finitely creased curves that only intersect creases in $S$. We need the following to ensure that, if we define  a flat folding of an infinite crease pattern step by step, we can use the limiting case of a flat folding for the entire crease pattern.

\begin{lemma}
\label{lem:local2global}
A tame crease pattern with a local flat folding has a flat folding if and only if every nearby region has a flat folding.
\end{lemma}

\begin{proof}
In one direction, a flat folding of the whole crease pattern can be restricted to any nearby region.
In the other, the number of creases must be countable (because of the open neighborhoods around each crease) so we can enumerate the creases as $c_0,c_1,\dots$ and let $S_i=\{c_j\mid j\le i\}$. For each $i$, the local region of $S_i$ has a finite set of flat foldings. Define an infinite tree with finite vertex degree in which the nodes are flat foldings of local regions of of sets $S_i$ for some $i$, and the parent of a node is the restriction of its flat folding to $S_{i-1}$. By K\H{o}nig's infinity lemma~\cite{Kon-ASM-27}, this tree has an infinite path. Any such path determines above--below relations of all points that cannot violate the conditions of \cref{lem:uncrossed}, because otherwise, these relations would also violate the same conditions in a nearby region of the creases involved in the violation.
\end{proof}

\subsection{Review of recent work}
\label{sec:review}

\begin{figure}[t]
\centering\includegraphics[width=0.4\textwidth]{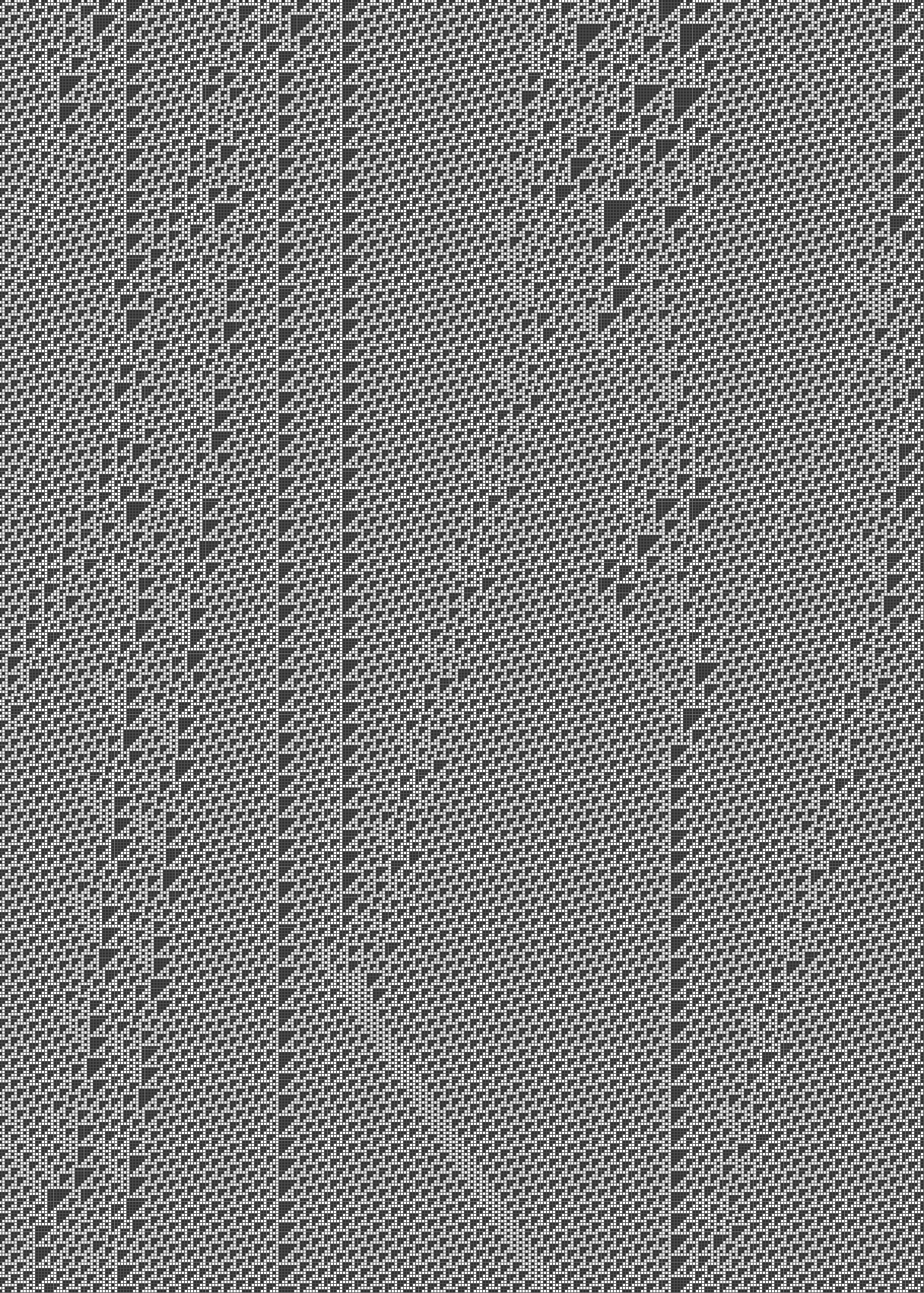}
\caption{Time-space diagram of the behavior of the Rule 110 cellular automaton on a random initial configuration (the top row of pixels of the figure). The $i$th row of the figure gives the state of the automaton after $i$ time steps. Several complex stationary and moving patterns are evident, overlaid on a more finely grained repeating background texture.}
\label{fig:rule110}
\end{figure}

Hull and Zakharevich~\cite{HullZak-23} use origami crease patterns to simulate the Rule 110 cellular automaton (\cref{fig:rule110}). The creases for their simulation repeat in periodic units on a half-plane, aligned with a triangular grid, and are labeled as mountain or valley and as required or optional folds. Required folds must be folded as specified; optional folds may remain unfolded. In each repeating unit, gadgets resembling those of Akitaya et al.~\cite{AkiCheDem-JCGCGG-15} simulate a Boolean circuit that simulates a single Rule 110 cell. Boolean signals enter the unit conveying the states of the cell and its neighbors in the previous time step. Boolean gates within the unit compute the state of the cell in the current time step, which is then sent out to the units for the same cell and its neighbors in the next time step.

Rule 110 was shown capable of universal computation by Matthew Cook~\cite{Coo-CS-04}. Cook uses Rule 110 to simulate Post tag systems, which repeatedly transform a string by removing a symbol at its start and concatenating a replacement string at its end. These tag systems in turn can simulate Turing machines, and it is undecidable whether a given tag system and initial state eventually leads to an empty string. Translating this result into the language of cellular automata, Cook proved that it is undecidable whether a Rule 110 system, with a given initial state that is periodic except for a finite perturbation, eventually becomes purely periodic.

Hull and Zakharevich encode the initial state of the Rule 110 pattern, in their simulation of Rule 110, by overlaying their crease pattern with additional required folds (in the position of formerly optional folds), on the input signals of the repeating units that represent the initial state of Rule 110. These additional required folds are (like the initial state of Cook) repeating, at a scale larger than a single cell, and with a finite perturbation. Hull and Zakharevich state only that their method can simulate universal computation, but one can make a more concrete statement of undecidability, translated from Cook's statement. If the Post tag system underlying Cook's result terminates in an empty string, the crease pattern of Hull and Zakharevich will produce a flat-folded state that is periodic for all but finitely many folds. If the tag system does not terminate, then the initial finite permutation will produce a folded state that differs from its periodic background at an infinite set of positions. Therefore, it is undecidable to determine whether a periodic infinite crease pattern on a half-plane, with optional folds, and with an overlay of required folds at its boundary that is periodic with a finite perturbation, can be folded in such a way that the perturbation remains finite. Thus, in a certain sense, folding infinite crease patterns is undecidable, but the undecidable problem is more complicated than flat foldability. The patterns of Hull and Zakharevich are always flat foldable. As well, their optional folds go beyond more common variants of the flat folding problem.

An alternative and independently developed method for proving Turing universality of origami folding has been recently suggested by Michael Assis~\cite{Ass-24}. Assis uses gadgets from Bern and Hayes~\cite{BerHay-SODA-96}, avoiding the bugs in the proof of Bern and Hayes by taking care that signal-carrying wire gadgets cross each other only perpendicularly. With this basis, his construction does not need optional folds. With alternative choices of gadgets it can use either unlabeled or labeled crease patterns. Assis observes that the gadgets of Bern and Hayes can form arbitrary Boolean logic circuits. He envisions from this a conventional CPU and memory built from these gates, connected to each other as sequential logic (incorporating feedback loops), with a sequence of folded states that change in  over time like the changing signals in the circuitry of an electronic computer, and kept in synchrony by a clock signal whose nature is not specified. Many details, such as the geometric layout of the resulting crease pattern, what it means for a folding to change dynamically, and how to incorporate an infinite memory into a conventional random-access machine design, are left unspecified. Therefore, it is unclear how to extract from his description a well-specified and undecidable decision problem.

In a talk presenting his work at the 8th International Meeting on Origami in Science, Mathematics and Education (8OSME) in Melbourne, 2024, Assis remarked that a universal computer constructed on this basis would (like the construction of Hull and Zakharevich) require an infinite sheet of paper, and a similar remark appears in his preprint: ``as long as we have infinite sheet of paper with an infinite number of NAND gates connected properly with an infinite amount of memory, we will have constructed a Universal Turing Machine''. Attempting to determine whether this requirement of infinite paper is truly necessary, in light of Tomoko Fuse's work on fractal folds on finite paper (presented on the same day of 8OSME),  became the basis for the new results of this section.

\subsection{Binary tiling}

\begin{figure}[t]
\centering
\includegraphics[width=0.4\textwidth]{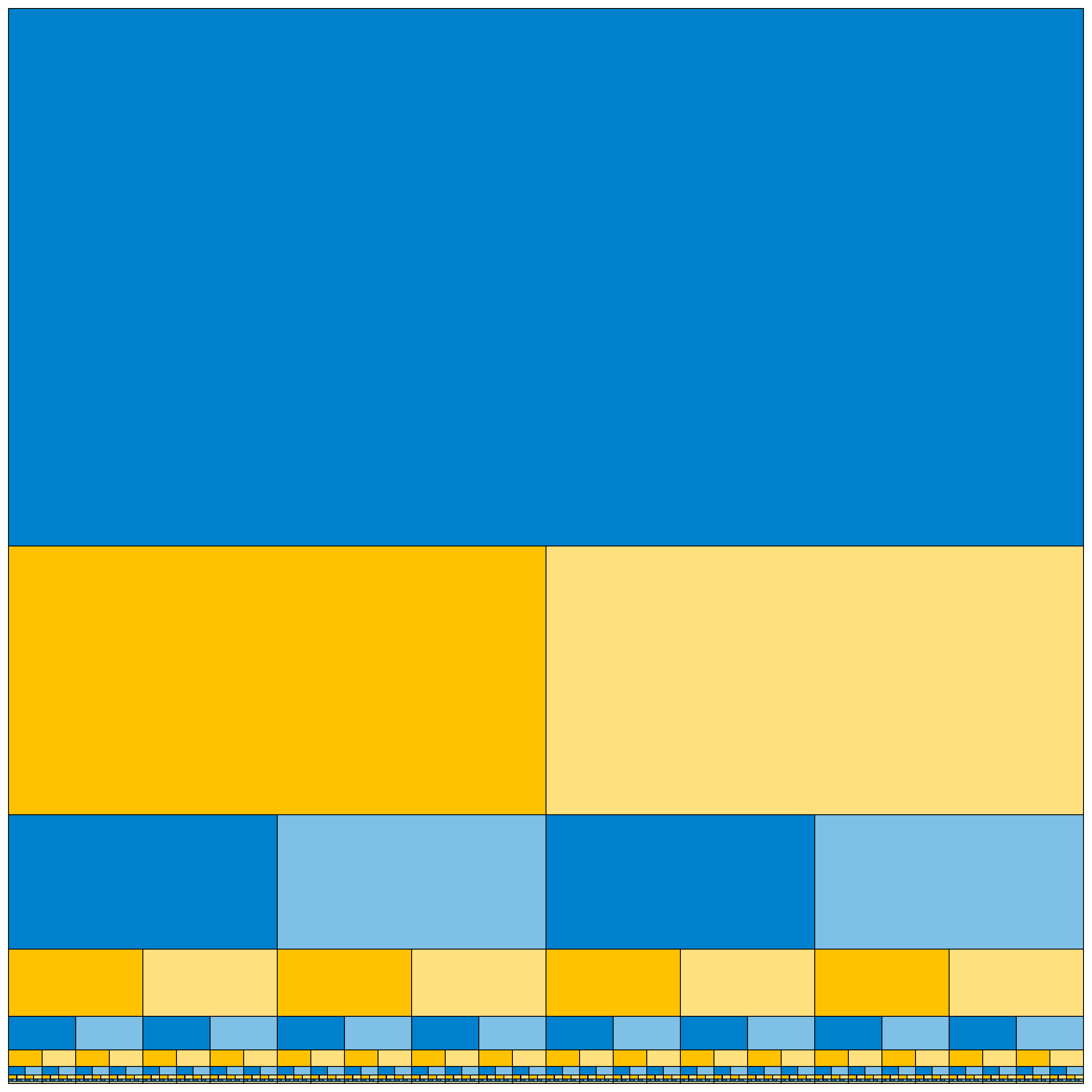}
\qquad
\includegraphics[width=0.4\textwidth]{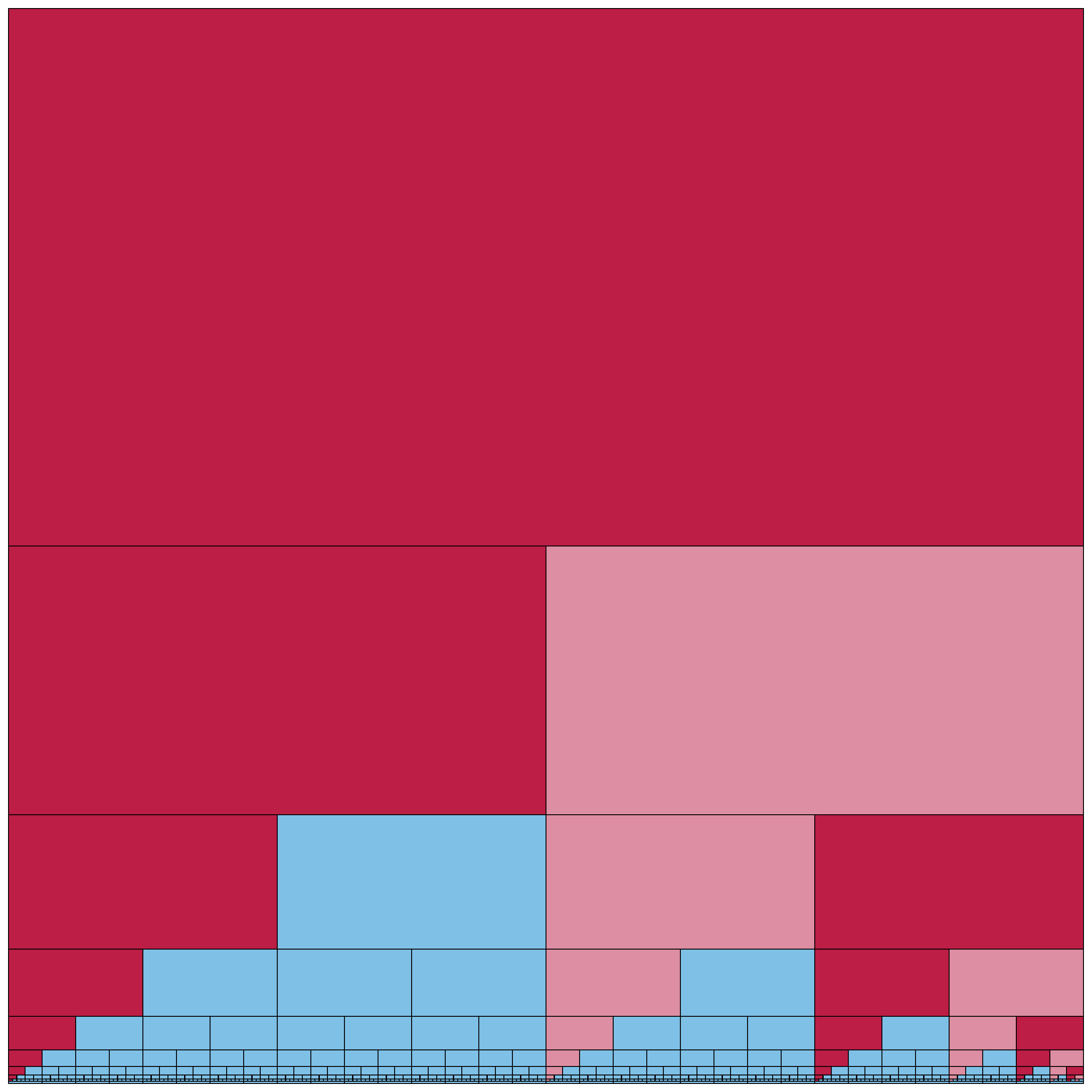}
\caption{The binary tiling, in the form of a recursive subdivision of a square into $2\times 1$ rectangles. In the coloring on the right, the alternating dark and light red columns of rectangles represent successive states of a Turing machine simulation; the blue rectangles are unused by the simulation, and will instead transfer a Turing machine cell state unchanged from left to right.}
\label{fig:binary}
\end{figure}

We simulate a Turing machine using a finite sheet of paper whose folding layout uses the \emph{binary tiling} in the form depicted in \cref{fig:binary}: a recursive self-similar subdivision of a square into $2\times 1$ rectangles, one covering the upper half-square, with the square's lower left and lower right quadrants each subdivided in the same pattern as the whole square. Each rectangle has two children, left and right, adjacent to its bottom edge, like the left and right children of a binary tree, and (except for the topmost rectangle) a parent rectangle above it. Additionally, each rectangle has left and right neighbors adjacent to it along its left and right sides.

The binary tiling was studied as a tiling of the hyperbolic plane by Böröczky in 1974~\cite{Bor-ML-74}. The outer square can be interpreted as part of a Poincaré half-plane model of the hyperbolic plane, with the square's bottom side lying on the line at infinity. The rectangular tiles model congruent non-convex hyperbolic quadrilaterals whose horizontal sides form arcs of horocycles; their vertical sides are straight line segments. Similar subdivisions of a disk by circular arcs (which can be interpreted in the same way as straight line segments and arcs of horocycles in the Poincaré disk model of the hyperbolic plane~\cite{Gup-IMM-06}) have been studied since the late 1920s as the Smith chart of radio engineering~\cite{Miz-JIECEJ-37,Smi-Electronics-39,Vol-Pb-40}. The binary tiling also models the interlocking lizards of M. C. Escher's \emph{Regular Division of the Plane VI} (1957)~\cite{Esc-EoE-89}. We will use its recursive self-similar structure as a decomposition of an origami square into rectangles, but we do not need its hyperbolic geometry.

We follow the following strategy for constructing a crease pattern that can simulate Turing--universal computation:
\begin{itemize}
\item Our crease pattern will repeat the same finite finite pattern of creases in each $2\times 1$ rectangle of the binary tiling, scaled to that rectangle. In this way it will have finite description complexity, suitable as input to a computational problem that we will prove undecidable.
\item As we detail in \cref{sec:circuits}, we design each rectangle's crease pattern to simulate a Boolean circuit, taking input binary signals from neighboring rectangles and passing outputs to neighboring rectangles in an acyclic signal-passing pattern.
\item We design a circuit within each rectangle that chooses between two different behaviors, which we call \emph{active} and \emph{passive}, in a pattern depicted in \cref{fig:binary} (right). Among active rectangles we distinguish the \emph{top} rectangles, topmost in a column of active rectangles.
\item To form this pattern of active and passive rectangles, each active rectangle passes a signal to its left child telling it to become active but not top. An active top rectangle signals its right child to become active and top. Non-top active rectangles instead signal their right children to become passive. Passive rectangles signal both children to become passive. In this way, signaling the topmost rectangle of the binary tiling to become active and top produces exactly the pattern of active and passive rectangles in \cref{fig:binary}.
\item We simulate a Turing machine with a single half-infinite tape and a single head somewhere on this tape. Each column of active rectangles simulates a single time step of this machine. Each rectangle in the column simulates a cell of the tape, including its symbol, whether it contains the Turing machine head, and its state if it does.
\item To simulate a time step of the Turing machine, an active rectangle takes in signals from its left side describing the previous (or initial) symbol on its tape cell, exchanges signals with its parent and left child sharing their symbols, and calculates the state that should result for the same tape cell in the next time step. These updated states are then routed down one position within the column, out the left sides of the non-top rectangles in the column, and across passive cells to the next column in the sequence of columns. The downward motion incorporated into this routing maintains the alignment of the half-infinite Turing machine tape cells with rectangles in columns of active rectangles.
\item The circuits simulated by the folds in each rectangle will determine the states, alphabet, and transition function of an arbitrary Turing machine. We will simulate a universal Turing machine, with additional constraints detailed later. As in the construction of Hull and Zakharevich~\cite{HullZak-23}, we augment the crease pattern with a finite amount of information about how to fold it, specifying that the topmost rectangle is top and active, that it contains the Turing machine head in its initial state, and specifying the symbols on finitely many tape cells. The constraints on the Turing machine will ensure that it behaves correctly with only this finite specification.
\end{itemize}

\subsection{Circuits in binary tiles}
\label{sec:circuits}

Our simulations of circuits by crease patterns use the gadgets from the $\mathsf{NP}$-completeness proof for flat folding by Akitaya et al.~\cite{AkiCheDem-JCGCGG-15}. Akitaya et al. provided two sets of gadgets, one for unlabeled crease patterns and a second for crease patterns labeled with mountain and valley folds. Both sets have similar geometry, with all creases horizontal, vertical, or at $45^\circ$ angles to the coordinate axes. They can be described schematically as colored vertices and edges in a directed graph, drawn in the plane in the rough positions of the folds of the corresponding crease pattern, with each directed edge drawn as a line segment that is again horizontal, vertical, or at $45^\circ$ angles.
\begin{itemize}
\item A variable gadget is drawn as a directed edge, and carries a Boolean value between its two endpoints. The direction of this edge may be unrelated to the direction of information flow through the circuit. Instead, it is a notational convention, used to define a correspondence between folding patterns and Boolean values. For unlabeled crease patterns, an edge gadget is represented by two parallel creases, which may be pleated in either of two ways. For labeled crease patterns, it is represented by four parallel creases, again with two foldings. In both sets of gadgets, reversing the direction of a variable gadget leaves its creases unchanged but inverts the correspondence between its two foldings and the Boolean values that it represents.
\item A split gadget is drawn as a white vertex where three variable gadgets meet: one with an incoming directed edge, and two with outgoing directed edges both at $135^\circ$ angles to the incoming edge. It can only be folded when these three edges all represent equal Boolean values. Akitaya et al. also present a second of split gadget with one outgoing edge and two incoming edges at $135^\circ$ angles to it, but it is really just the same crease pattern reinterpreted under edge reversals and complementation of the corresponding Boolean values. We may allow split gadgets with other choices of incoming and outgoing edges, always with the same crease pattern and with the allowed Boolean values complemented for each reversed edge.
\item A clause gadget is drawn as a black vertex where three variable gadgets meet, at the same angles as a split gadget: two edges at $135^\circ$ angles to a third edge, and at $90^\circ$ angles to each other. When the three variable gadgets are oriented all incoming or all outgoing, the clause gadget can only be folded when the Boolean values that they represent are not all equal.
\item A cross gadget is drawn by Akitaya et al. as a crossed white vertex, at which four variable gadgets meet at right angles, but we will draw it as a nonplanar crossing of two variable gadgets that each continue in both directions through the gadget. In either style of drawing the effect is the same: the Boolean values represented on the vertical segments above and below the gadget must be equal, the Boolean values represented on the horizontal segments to the left and right of the gadget must be equal, and the vertical and horizontal Boolean values may be chosen independently of each other.
\end{itemize}

\begin{figure}[t]
\centering\includegraphics[scale=0.3]{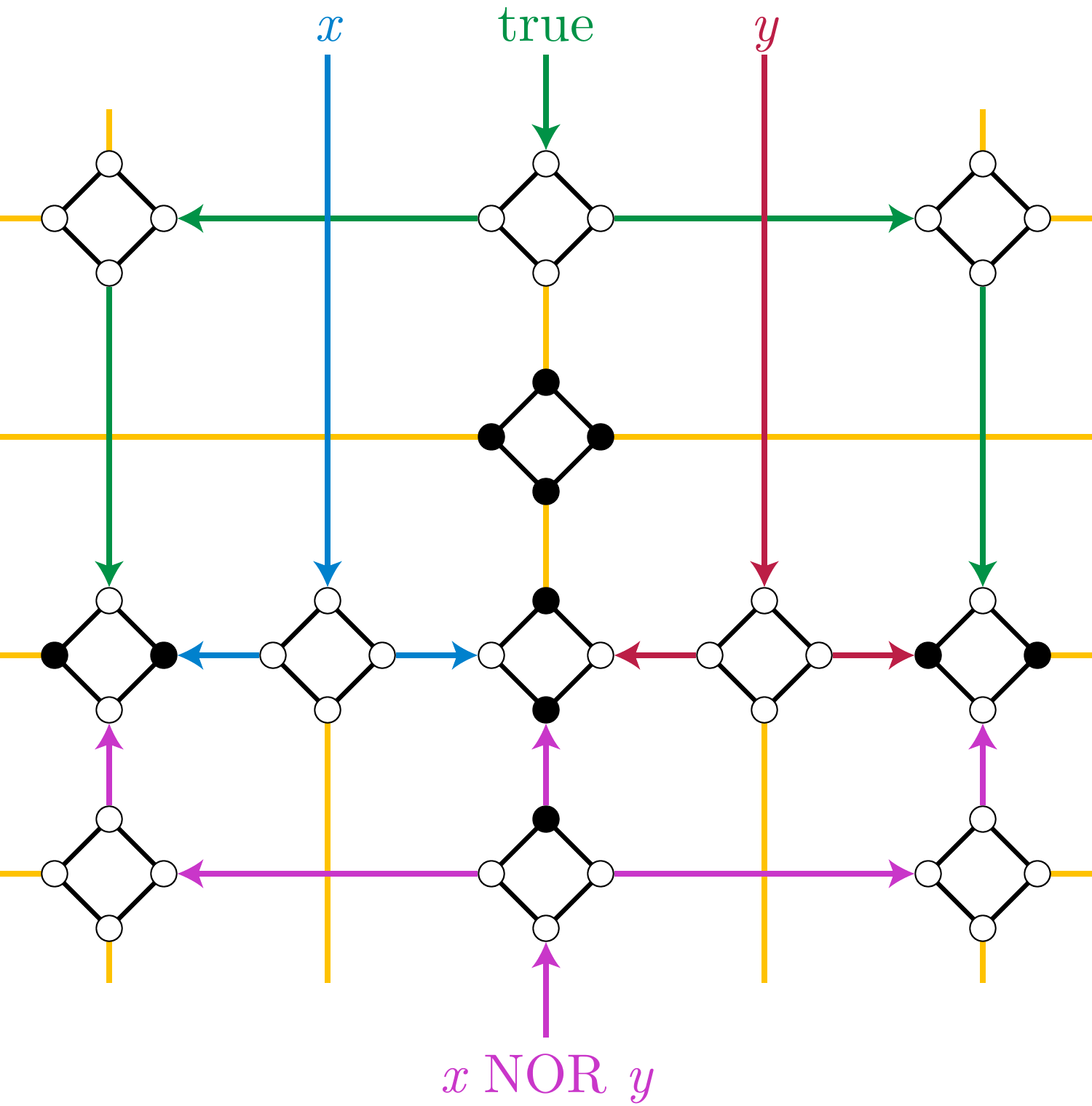}
\caption{NOR gate built from the gadgets of Akitaya et al.~\cite{AkiCheDem-JCGCGG-15}, using our grid-based layout convention. The output of the NOR gate is depicted with an inward-pointing arrow: the direction of each arrow may be unrelated to the direction of information flow. Yellow edges indicate variable gadgets whose value is unused or arbitrary; we leave these edges undirected.}
\label{fig:grid-nor}
\end{figure}

We adopt the following strict layout convention to build circuits out of these gadgets (\cref{fig:grid-nor}):
\begin{itemize}
\item For an appropriate parameter $k$, we will divide each rectangle of the binary tiling into a grid of $2k\times k$ squares, each having variable gadgets crossing the midpoints of its four sides.
\item Variable gadgets may carry one of three classes of values:
\begin{enumerate}
\item They may carry values intended to be used as inputs, outputs, or intermediate values of a circuit, such as the variables labeled $x$, $y$, and $x$ NOR $y$ in the figure.
\item They may be unused, either waste values produced from parts of the circuit, or filler used to space the components of the circuit in their intended positions in the grid. These may carry arbitrary values and the circuit should function as intended regardless of these values. In the figure, these are the variable gadgets depicted as yellow edges.
\item To distinguish true from false Boolean values, we connect throughout the grid (in all rectangles of our binary tiling) variable gadgets carrying a Boolean value ``true'' (green in the figure). A signal-carrying variable gadget will be interpreted as being true if it has the same folding pattern (for the same arrow orientation) as the ``true'' variable gadgets, and false otherwise. We do this because our gadgets act symmetrically under true--false reversal and would otherwise be unable to compute Boolean functions that are not likewise symmetric. A similar convention was used for the same reason by Assis~\cite{Ass-24}.
\end{enumerate}
\item Each grid square will contain either a single cross gadget (allowing the variable gadgets to connect from left to right and top to bottom), or four split and clause gadgets, arranged in a diamond (tilted square), with variable gadgets on its sides. Each of these four gadgets terminates a horizontal or vertical variable gadget entering one midpoint of the grid square.
\end{itemize}

Different diamonds of four elementary gadgets in a grid square have different effects on the Boolean values on their incident horizontal and vertical variable gadgets:
\begin{itemize}
\item A diamond of four clause gadgets (black vertices), such as the one in the upper center of \cref{fig:grid-nor}, allows any combination of Boolean values on the incident horizontal and vertical variable gadgets. The Boolean values on the sides of the diamond can alternate true and false, satisfying the clause gadgets regardless of the other variable gadgets. This allows, for instance, unused or waste signals to be terminated without interfering with each other. Each variable gadget intended to carry an unused signal should have a grid square of this type at one or both of its endpoints.
\item A diamond of four split gadgets, or three split gadgets and one clause gadget (both visible in \cref{fig:grid-nor}) forces all four of its incident variable gadgets to have the same Boolean value (or the complementary value, depending on arrow direction). Both variations are equivalent except for the conventional direction of their arrows. This can also be used to complement a Boolean value, relative to the value that would result from using a crossover gadget, while producing two waste values in the perpendicular directions.
\item A diamond of two opposite split gadgets and two opposite clause gadgets has the effect of applying the two clause gadgets to two triples of incoming values. There are three of these in \cref{fig:grid-nor}; in each of the four, one of the two clause gadgets is an essential part of the circuit while the other one produces a waste bit as output.
\end{itemize}

Grid squares containing a cross gadget, four clause gadgets, or four split gadgets are sufficient to allow arbitrary Boolean values (including the ``true'' value) to be routed to any part of a circuit, insulated from each other by grid squares with four clause gadgets, and inverted or not as required. The remaining grid squares, which apply clause gadgets to triples of values, can be used to implement a Boolean NOR gate (\cref{fig:grid-nor}) and hence any other Boolean gate or acyclic circuit of Boolean gates. The NOR gate of the figure operates by routing its two inputs $x$ and $y$, and the ``true'' value, to three clause gadgets, and then using split gadgets to force the three ``output'' signals from these clause gadgets to have a single shared Boolean value. It resembles in structure, if not in its geometric layout, a similar construction by Assis~\cite{Ass-24} for a NAND gate based on the gadgets of Bern and Hayes~\cite{BerHay-SODA-96}. One may verify using truth tables that the only valid foldings of the figure are ones in which the $x$ NOR $y$ output signal carries the Boolean value $x$ NOR $y$. Any schematic representation of these gadgets, laid out according to this layout convention, can be converted into a crease pattern, with the gadgets all having the same geometric proportions, chosen small enough to fit into the grid squares.

One detail that we will need, going beyond pure combinational logic, is a mechanism for preventing a crease pattern from folding flat, based on the value of a logical signal (analogous to the mythical ``halt and catch fire'' instruction of some computer architectures). Doing so is straightforward: route the signal to a clause gadget, together with two copies of the ``true'' signal. If the signal is also true, the clause gadget will fail to fold. We can think of this as emulating a special kind of Boolean logic gate with a single input and no outputs, a ``halt gate''.

We summarize the observations of this section:

\begin{observation}
Any combinational (acyclic) Boolean logic circuit (allowing halt gates) can be simulated using gadgets of the types described by Akitaya et al.~\cite{AkiCheDem-JCGCGG-15}, laid out schematically in the grid conventions described above, with the circuit's inputs and outputs connected to arbitrarily chosen grid square midpoints on the boundary of a rectangle of $2k\times k$ grid squares, for any sufficiently large $k$.
\end{observation}

We do not yet write that we can simulate the circuit by a finite crease pattern using the crease pattern designs by Akitaya et al., nor yet that we can connect signals from one rectangle to another, because of the scale and alignment issues that we address in the next two sections.

\subsection{Scale}

Although we have described the gadgets of Akitaya et al.~\cite{AkiCheDem-JCGCGG-15} as 0-dimensional point vertices and 1-dimensional line segment edges, their actual design is two-dimensional, as fold patterns occupying an area within a crease pattern. This area can be parameterized by the \emph{line width} of the variable gadgets. In defining width, it is equivalent to consider the distance between the outer parallel creases of an unfolded variable gadget, the distance between its outer creases in its folded state, or the amount by which folding this gadget reduces the distance between points on opposite sides of the gadget; these definitions are equivalent up to fixed factors which depend on the gadget design but not its scale. We use the third definition, the amount by which folding reduces distances.

A single variable gadget can be constructed at any scale that fits into the paper to be folded, but the scales of interacting variable gadgets are not independent of each other. For the crease patterns of Akitaya et al.~\cite{AkiCheDem-JCGCGG-15}, two variable gadgets that cross at a cross gadget have the same width. At the split and clause gadgets, the three incident variable gadgets have two different widths: a smaller width for the two variables that meet at a right angle, and a larger width for the third variable, at a $135^\circ$ angle to the other two. One can calculate the ratio of these widths by examining the split and clause crease patterns in both labeled and unlabeled forms, but instead these width ratios can be determined theoretically in a unified way:

\begin{figure}[t]
\centering\includegraphics[scale=0.23]{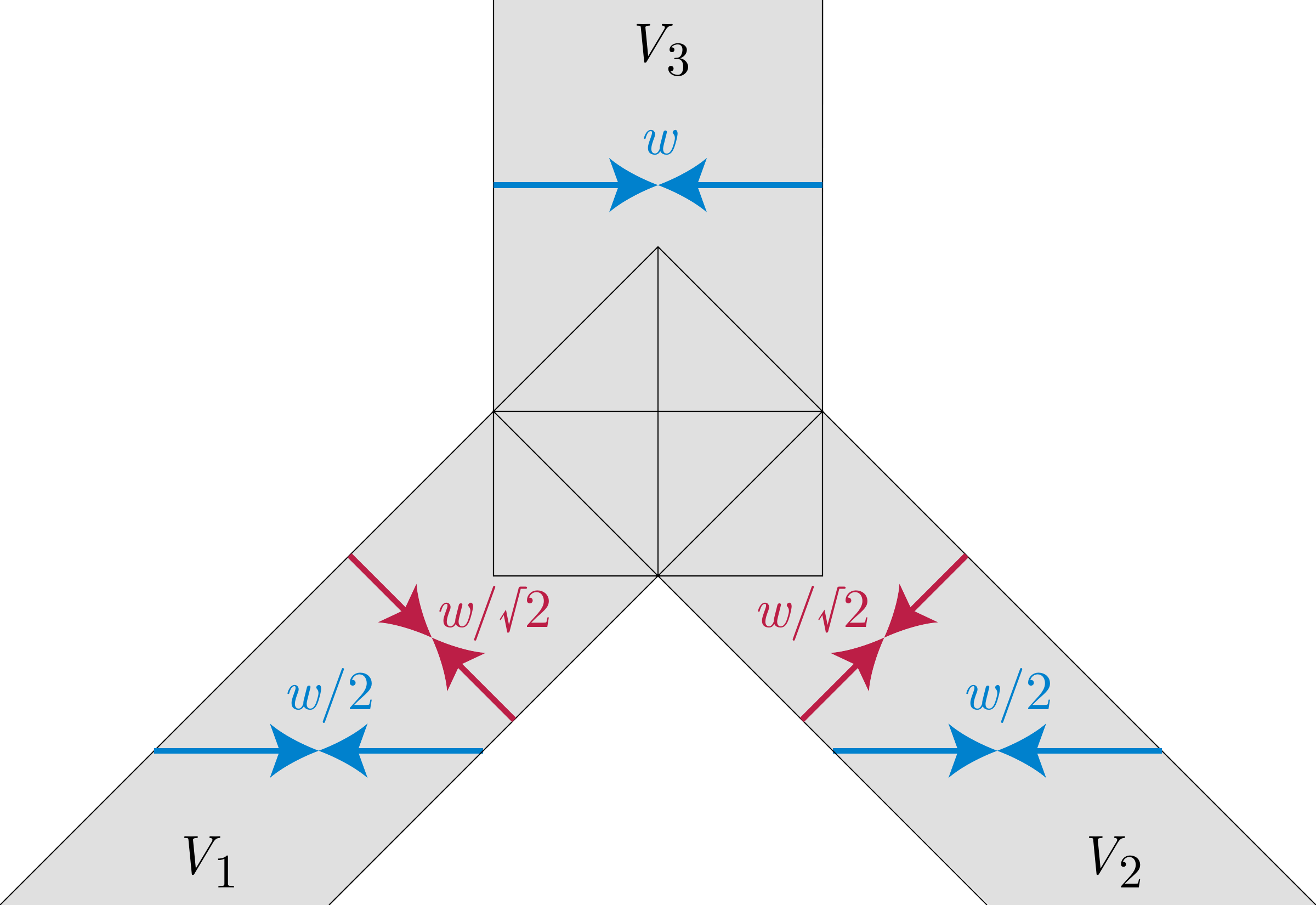}
\caption{Notation for \cref{lem:varscale}: three variable gadgets, meeting at a split or clause gadget (shown schematically rather than by its crease pattern), labeled by how much folding reduces the horizontal distance (blue) or perpendicular distance (red) between points on opposite sides of the variables.}
\label{fig:scale}
\end{figure}

\begin{lemma}
\label{lem:varscale}
At any gadget where three variable gadgets meet, two at right angles to each other and a third at $135^\circ$ angles to both of them, and in which the two right-angled variable gadgets have equal widths, this width is smaller from the width of the third variable gadget by a factor of $1/\sqrt2$.
\end{lemma}

\begin{proof}
As a notational convention, align the gadgets so that the third variable gadget $V_3$ (the one making $135^\circ$ angles to the other two variable gadgets $V_1$ and $V_2$) is vertical, with width $w$ (see \cref{fig:scale}).
Thus, by our definition of width, folding these gadgets brings points in the plane that are on the left side of $V_3$ closer to their reflections across $V_3$ on its right side, reducing their distance by $w$.
The left side and right sides of $V_3$ are bounded by two $135^\circ$ uncreased wedges in the plane, which extend downwards of  the gadget $G$ at which $V_1$, $V_2$, and $V_3$ meet; thus, for points within these wedges, even below $G$, a point in the left wedge is brought closer to its reflection in the right wedge by a distance of $w$, by folding the gadgets. By the assumption that $V_1$ and $V_2$ have the same width, the midpoint of these two reflected points is brought closer to both points (considering only horizontal distance, i.e. the absolute difference of $x$-coordinates) by $w/2$ when the gadgets are folded. (This reduction in horizontal distance is not the same as the horizontal extent of gadgets $V_1$ and $V_2$.) Because $V_1$ and $V_2$ are oriented diagonally, at slopes $\pm 1$, they must have width $w/\sqrt2$ in order to make the horizontal component of distance reduction be $w/2$.
\end{proof}

\begin{figure}[t]
\centering\includegraphics[scale=0.2]{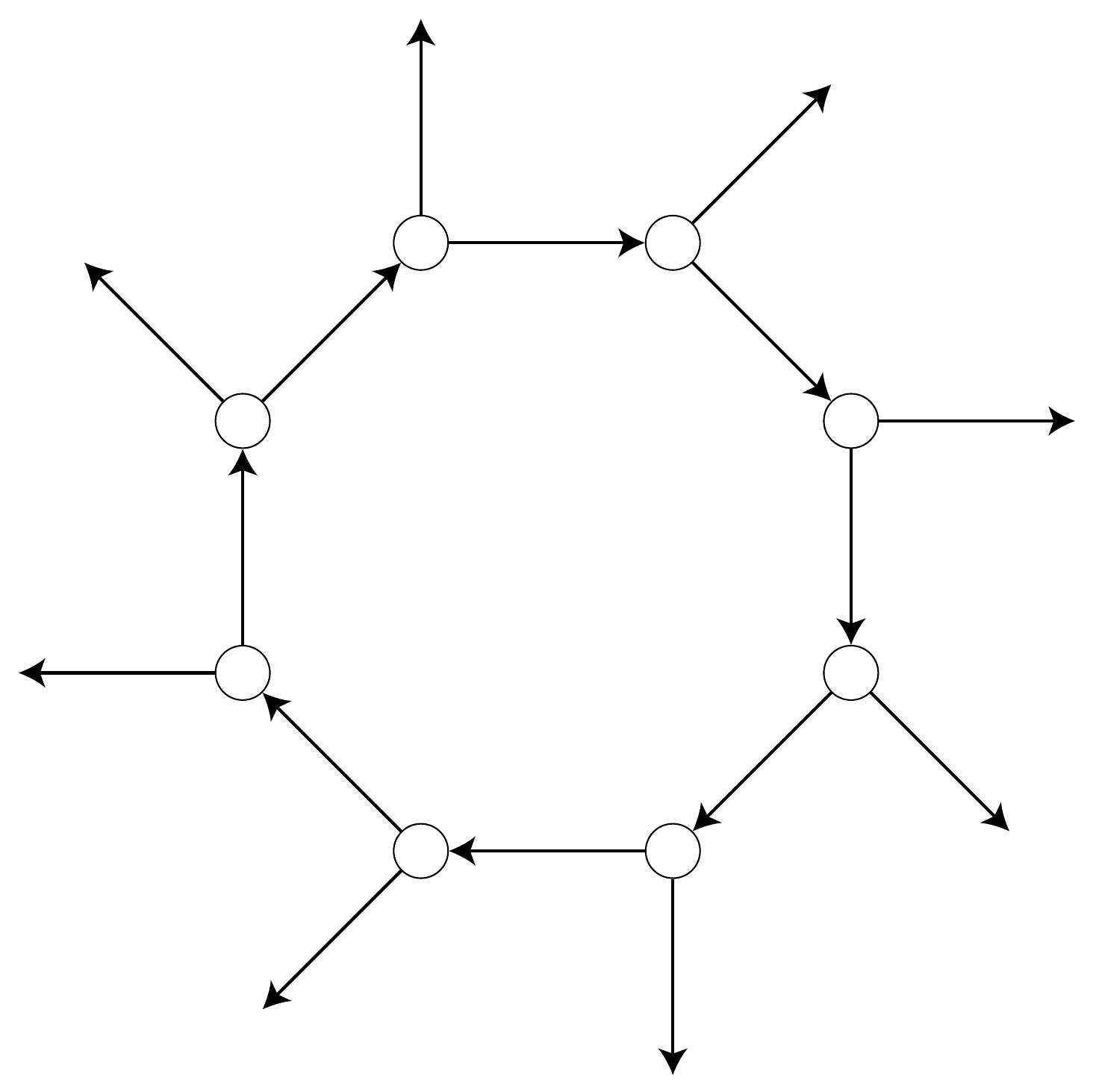}
\caption{A layout of variable and split gadgets with incompatible scales. Each variable gadget (directed edge) in the central cycle must have width $1/\sqrt2$ times its predecessor's width, resulting in a 16-fold reduction in width after eight steps around the cycle returning to the start, an impossibility.}
\label{fig:incompatible-scale}
\end{figure}

Because of this change of width, certain arrangements of gadgets have incompatible widths (\cref{fig:incompatible-scale}), and we must verify that our intended layout does not have such incompatibilities. Akitaya et al. do not describe or check this scale issue, but the layouts they use do not have incompatible scales. Perhaps for this reason Hull and Zakharevich~\cite{HullZak-23} instead designed their gadgets with variable gadgets meeting at $120^\circ$ angles, all the same width. However, for our purposes, scale change is a feature: it allows us to connect circuits of different scales in different rectangles of the binary tiling.

\begin{observation}
For the circuit layouts detailed in \cref{sec:circuits}, all variable gadgets can be designed with compatible scales at each cross, split, or clause gadget.
\end{observation}

\begin{proof}
We choose some width $w$, sufficiently small for each gadget to fit into the grid squares of the layout, use width $w$ for the vertical and horizontal variable gadgets, and use width $w/\sqrt2$ for the diagonal variable gadgets. All cross gadgets involve the crossing of a vertical and horizontal variable gadgets, of equal width $w$ and all split or clause gadget involve a single vertical or horizontal variable gadget of width $w$ at $135^\circ$ angles to the other two diagonal variable gadgets of widths $w/\sqrt2$, matching the restrictions of \cref{lem:varscale}.
\end{proof}

\begin{figure}[t]
\centering\includegraphics[scale=0.3]{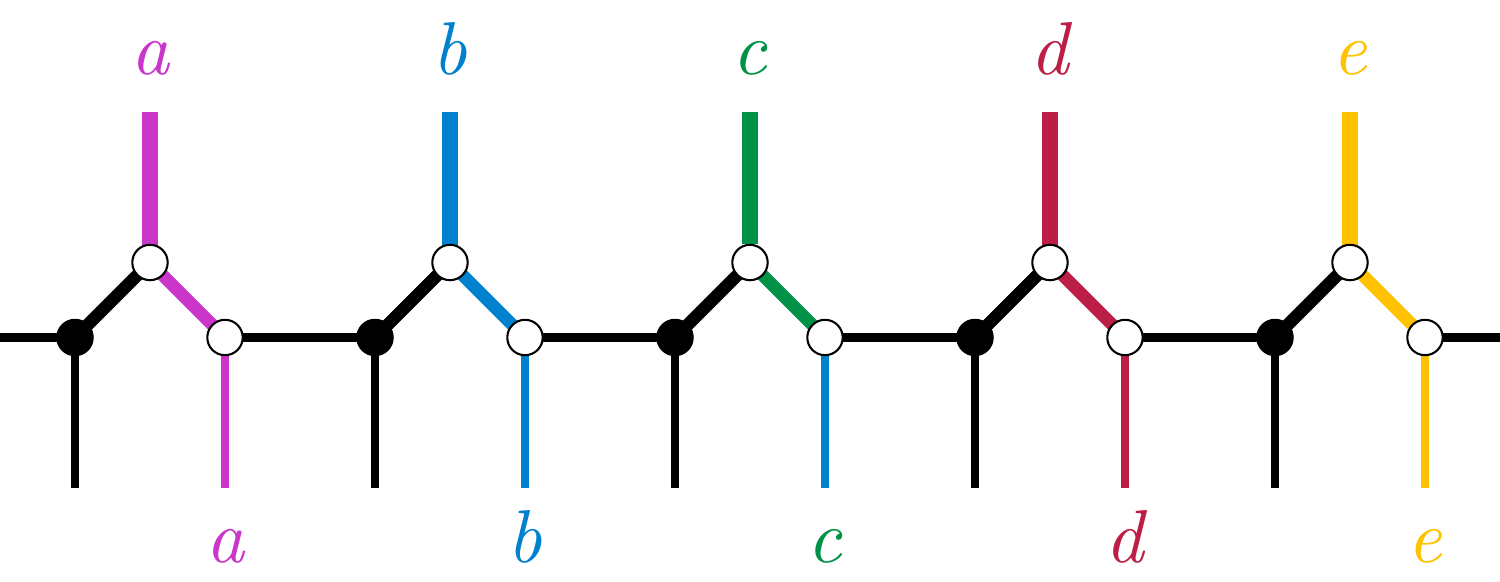}
\caption{Scale transition from variable gadgets in a parent rectangle with gadget width $w$ (the colored and labeled vertical line segments at the top of the figure) to variable gadgets in a child rectangle with gadget width $w/2$ (the corresponding colored and labeled vertical line segments at the bottom of the figure). The split gadgets in the central part of the figure (white circles) ensure that corresponding parent and child variable gadgets have the same value. The clause gadgets (black circles) separate these values from each other in the child rectangle by waste signals (black line segments). Variable widths are shown in approximately the correct proportions to each other.}
\label{fig:scale-transition}
\end{figure}

As discussed earlier, each rectangle of the binary tiling will have a scaled copy of the crease pattern used at the upper root rectangle. If we design the gadgets of this root crease pattern to fit in the root rectangle, their scaled copies will automatically fit in the same way into every other rectangle of the binary tiling. However, this leaves us with a difficulty: we need to send binary signals between a parent rectangle, with variable gadgets having some width $w$ scaled to its size, and its left and right children with the different width $w/2$. Variable gadgets cannot cross directly from a parent rectangle to a child rectangle, as they can between left and right neighbor rectangles, because after such a crossing the variable gadget would have the wrong scale.

Instead, we connect the variable gadgets in each grid column of a parent rectangle to variable gadgets in its child rectangles using the layout of gadgets shown in \cref{fig:scale-transition}. This layout ensures that corresponding Boolean values in the parent and child rectangles have the same value, regardless of which direction information flows. The black clause gadgets of the figure separate these values, producing waste values on alternating vertical paths into the child rectangle. This pattern can be extended from left to right across the entire origami square, with its horizontal black edges aligned with the midpoints of the topmost row of grid squares in each child rectangle.

\begin{observation}
The scale transition above preserves signal values across the boundary between parent and child rectangles, has compatible variable gadget widths at all split and clause gadgets, and is aligned with the grid squares of both the parent and child rectangles.
\end{observation}

It is for the purpose of this grid alignment that we lay out our circuits with variable gadgets passing through the midpoints of grid squares, rather than through grid edges.

\subsection{Alignment}

\begin{figure}[t]
\centering\includegraphics[width=0.4\textwidth]{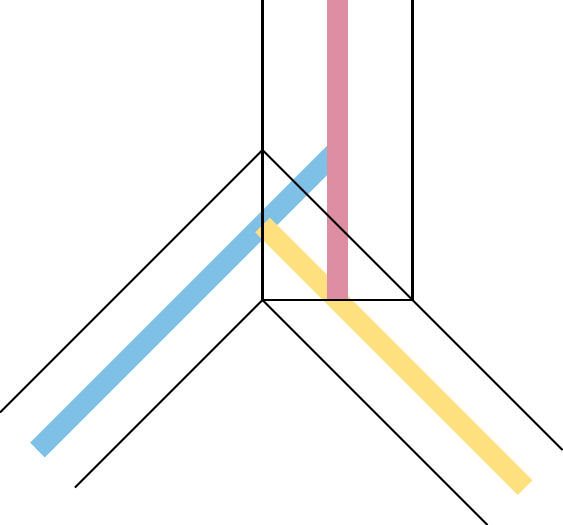}
\caption{For the unlabeled clause gadget of Akitaya et al.~\cite{AkiCheDem-JCGCGG-15} (black lines), the symmetry axes of its three incident variable gadgets (thick colored lines) do not meet at a point.}
\label{fig:unlabeled-clause-alignment}
\end{figure}

\begin{figure}[t]
\centering\includegraphics[scale=0.25]{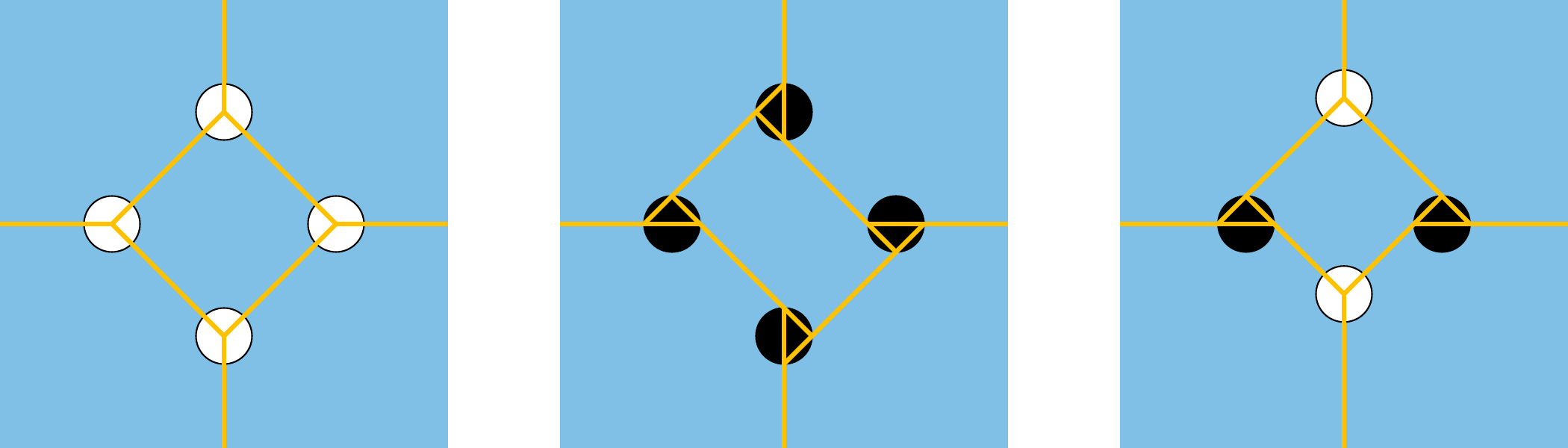}
\caption{Realignment of grid squares containing two or four unlabeled clause gadgets}
\label{fig:realigned}
\end{figure}

When we replace the schematic layout of our gadgets by their crease patterns, after choosing the variable width, we also need to determine consistent positions of the variable gadgets relative to their corresponding schematic line segments. The variable gadgets have crease patterns that are symmetric across a line parallel to the gadget, and ideally we would like to place each variable gadget with its axis of symmetry aligned with its schematic line segment.

Five of the six vertex gadgets of Akitaya et al., the unlabeled and labeled cross and split gadgets, and the labeled clause gadget, allow this symmetric alignment. But for the unlabeled clause gadget, the three symmetry axes of incident variable gadgets do not meet in a point (\cref{fig:unlabeled-clause-alignment}), so the schematic alignment cannot match the actual alignment. In a finite crease pattern, this misalignment merely causes small changes of position of other gadgets throughout the design. But in an infinite crease pattern such as ours with different scales for different parts of the pattern, it is important to maintain the same crease pattern in all scaled rectangles of the binary tiling, and to prevent misalignment of gadgets at one scale from propagating to bigger misalignments at a smaller scale.

For grid squares containing two or four clause gadgets, a suitable placement for these gadgets allows all the changes of alignment from these gadgets to affect the diagonal variable gadgets within the grid square; the vertical and horizontal variable gadgets connecting it to neighboring squares are unaffected (\cref{fig:realigned}). Grid squares containing a single clause gadget, such as the one used in \cref{fig:grid-nor}, are unnecessary as the same function may be performed by grid squares containing four split gadgets, and so they may be omitted from our design altogether.

\begin{figure}[t]
\centering\includegraphics[scale=0.25]{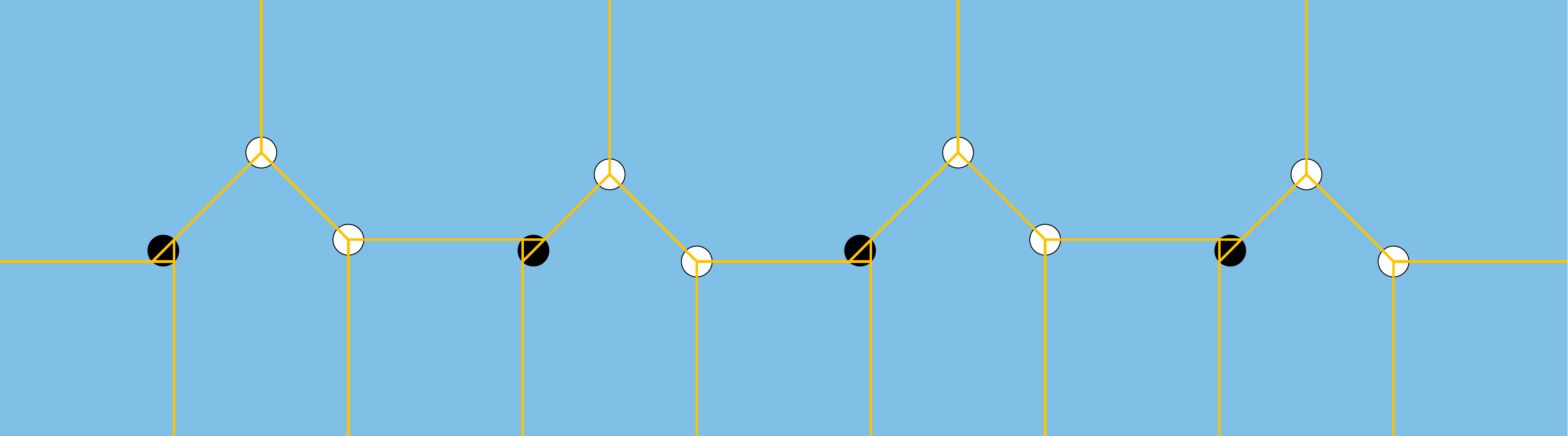}
\caption{Realignment of scale transition}
\label{fig:realigned-scale}
\end{figure}

For the transition between scales at the boundary between parent and child rectangles, we use clause gadgets alternatingly aligned above and below their incoming diagonal edge (\cref{fig:realigned-scale}), together with suitable vertical adjustments of the positions of split gadgets. The vertical variable gadgets in the two adjacent rectangles keep their exact alignments; only the diagonal and horizontal variable gadgets are realigned. Because each rectangle of the binary tiling is subdivided into a $2k\times k$ grid of squares, the number of alternations is even, causing this pattern to line up with the same pattern in the left and right neighbors of each rectangle.

\subsection{Universality}

Certain individual Turing machines, called \emph{universal Turing machines}, with a fixed alphabet, fixed state set, and fixed transition function, can simulate any other Turing machine, even one with more tapes, alphabet symbols, states. A natural way to design universal Turing machines is to have them directly simulate other Turing machines. Each alphabet symbol of the simulated machine can be represented by a sufficiently large tuple of tape cells of the simulating machine. The simulating machine can reserve part of its tape to store a representation of the states and state transitions of the simulated machine, using the rest of its tape to store its representation of the tape(s) of the simulated machine. Special alphabet symbols (or combinations of binary alphabet values) can be used as markers into the state transition table and simulated tape, allowing the simulating machine to bounce back and forth between the state transition table and simulated tape as it simulates one step of the simulated machine in many steps of the simulating machine. In this way we may to design a universal Turing machine $\mathcal U$ with the following additional properties, needed for our application:
\begin{itemize}
\item $\mathcal{U}$ operates on a semi-infinite tape. Finitely many contiguous tape cells provide input to $\mathcal U$, specifying the states, state transition function, alphabet, initial tape state, and initial head position of the Turing machine being simulated by $\mathcal U$. This finite input is terminated at each end by two distinct \emph{active tape terminators}.
\item $\mathcal{U}$ never changes, nor moves earlier in its tape than, the first active tape terminators. If it ever moves later in its tape than the second active tape terminator, it removes this terminator from its previous position and immediately places it at the new position, so the second position always records the farthest position to which the head has moved.
\item If the machine that is being simulated halts, $\mathcal{U}$ halts as well.
\end{itemize}
Because of the first active tape terminator, $\mathcal U$ will never move off the start of its semi-infinite tape, saving us from specifying what to do in that event. Because of the second active tape terminator, the state of the tape beyond these terminators is irrelevant to its behavior; it will behave the same regardless of what else is on the tape. The following complexity-theoretic result is standard:

\begin{lemma}[halting problem]
For any fixed universal Turing machine $\mathcal{U}$ that operates by direct simulation of an arbitrary Turing machine, let $\operatorname{Halt}_{\mathcal U}$ be the language of inputs that cause $\mathcal{U}$ to halt. Then $\operatorname{Halt}_{\mathcal U}$ is undecidable, and $\mathsf{RE}$-complete.
\end{lemma}

\begin{theorem}
\label{thm:undecidable}
There exists a fixed tame crease pattern $P$ for a square of origami paper (with either labeled or unlabeled creases) such that, for a given finite set $F$ of decisions about how to fold creases on the boundary of $P$, it is undecidable and $\mathsf{coRE}$-complete to determine whether $P$ can be flat folded consistently with $F$.
\end{theorem}

\begin{proof}
We apply the circuit design methodology detailed in earlier subsections to convert $\mathcal{U}$ to a crease pattern $P_\mathcal{U}$,
consisting of copies of the same finite crease pattern in each rectangle of a binary tiling. The foldings of $P_\mathcal{U}$ simulate $\mathcal{U}$ on an initial tape whose initial state (between and including the active tape terminators) can be specified by the above--below relation of finitely many folds at the boundary of $P_\mathcal{U}$. We also specify the binary signals that tell the root rectangle of $P_\mathcal{U}$ to be an active top rectangle, again using the above--below relation for finitely many folds  at the boundary of $P_\mathcal{U}$. If $\mathcal{U}$ halts on this input, $P_\mathcal{U}$ will fail to fold, by triggering a halt gate in the rectangle modeling the tape cell and time step at which it halts. If $\mathcal{U}$ does not halt on this input, then its non-halting run can be translated into a choice of how to fold $P_\mathcal{U}$ that puts valid tape cell values (rather than, say, extra Turing machine heads) on all unspecified input signals, and then folds each gadget in $P_\mathcal{U}$ according to the values of the Boolean circuit that it represents. This non-halting run will not trigger any halt gates, and because $P_\mathcal{U}$ is tame, this sequence of folding decisions will lead to a valid flat folding of $P_\mathcal{U}$ by \cref{lem:local2global}.

Thus, let $P$ be $P_\mathcal{U}$ and $F$ be the finite set of folding decisions that specifies the input tape to $\mathcal{U}$ and the requirement that the root rectangle of the binary tiling be active and top. Then $P$ can be flat folded consistently with $F$ if and only if $\mathcal{U}$ does not halt on the given input.
Because the halting problem for $\mathcal{U}$ on a given input is $\mathsf{RE}$-complete, its complementary problem (true if it does not halt) is  $\mathsf{coRE}$-complete. Thus, as claimed, testing whether $P$ can be flat folded consistently with $F$ is $\mathsf{coRE}$-complete and therefore also undecidable.
\end{proof}

\subsection{Wang tiling}

\begin{figure}
\centering\includegraphics[width=0.8\textwidth]{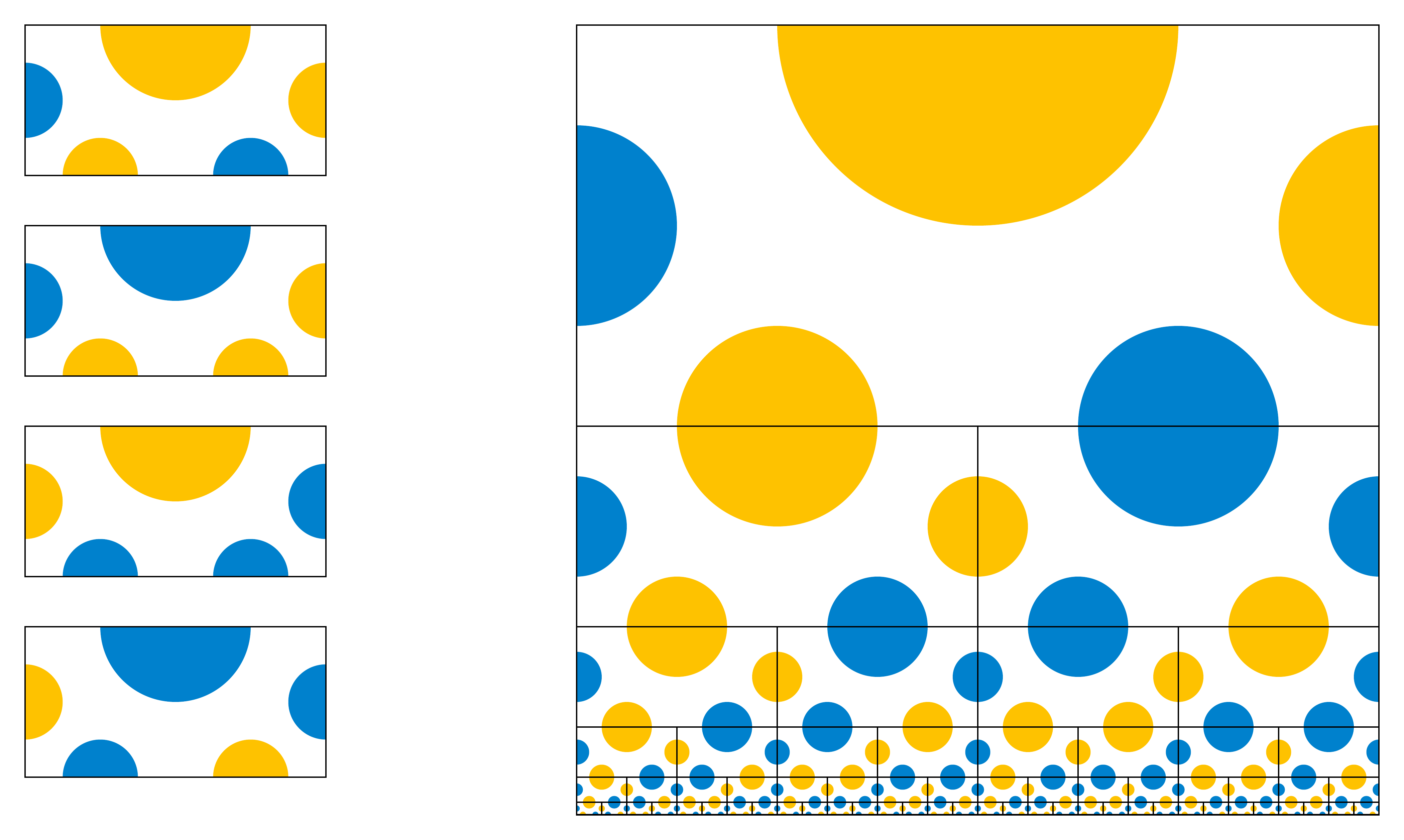}
\caption{The prototiles of a binary Wang tiling problem (left) and a binary Wang tiling using them (right). The semicircles in each tile color its five edges.
The edge coloring along each horizontal line of the tiling follows the alternate paper-folding (dragon curve) sequence~\cite[section 4]{DavKnu-JRM-70}.}
\label{fig:binary-wang}
\end{figure}

In this section we describe an alternative proof of undecidability to the one in \cref{thm:undecidable}, based on the same idea of embedding circuits in a binary tiling, for which the undecidable instances consist only of a crease pattern with the same creases in each rectangle of a binary tiling. Unlike \cref{thm:undecidable} we do not also include a finite set of folding decisions within the instances. However, in order to provide inputs of unbounded size, we allow the crease pattern to vary. The basis for our proof is a generalization of the \emph{Wang tiling} problem to binary tilings, studied by Jarkko Kari~\cite{Kar-MCU-07}. Following Kari, we define an input to the \emph{binary Wang tiling problem} to consist of a finite set of \emph{prototiles}, $2\times 1$ rectangles with colored edges. We consider the bottom side of a $2\times 1$ rectangle to consist of two edges, so that the tiles of the binary tiling meet edge-to-edge. Given this data, a binary Wang tiling is simply a copy of the binary tiling, in which each tile uses the coloring from one of the prototiles, and such that each two tiles that share an edge use the same color for that edge. The output of the problem is true if there exists such a tiling, and false otherwise. A true instance with four colored tiles and an example of a binary Wang tiling for this instance are depicted in \cref{fig:binary-wang}.

Kari considers binary tilings in the hyperbolic plane rather than in a square but this turns out not to be of significance for the binary Wang tiling problem.
The hyperbolic plane may be modeled by a half-plane, with a horizontal boundary (the Poincaré half-plane model) for which the hyperbolic isometries include left-to-right translation and scaling. If our square binary tiling is placed into this half-plane, with its lower edge on the half-plane boundary, all the tiles are isometric, with their boundaries consisting of vertical hyperbolic line segments and horizontal arcs of hyperbolic horocycles. The tiling within the square may be extended (in uncountably many distinct ways) into a tiling of the entire half-plane, with tiles of the same shape meeting each other in the same way; these are the hyperbolic binary tilings.

\begin{lemma}
A given instance of the binary Wang tiling problem can tile the square, if and only if it can tile every binary tiling of the Poincaré half-plane, if and only if it can tile every finite union of tiles of either the square binary tiling or of the half-plane binary tiling.
\end{lemma}

\begin{proof}
This is a standard application of K\H{o}nig's infinity lemma~\cite{Kon-ASM-27}. If there exists a Wang tiling for the entire half-plane, there exists a Wang tiling for its subset, the square binary tiling. If there exists a Wang tiling for the square, the tiling can be restricted to any finite union of tiles. And if any finite union of tiles in the square can be tiled, then so can any finite union of tiles of the half-plane, by scaling them down to a copy that fits within the square and using the Wang tiling of that copy.

Thus, it remains to show that, if every finite union of binary tiles of the half-plane has a Wang tiling, then so does every binary tiling of the half-plane. Choose an arbitrary tiling of the half-plane, enumerate its tiles, and consider the unions of prefixes of this enumeration. Form an infinite tree whose nodes are the Wang tilings of each such union, with the parent of any node obtained by removing the final tile from its prefix. This tree has bounded degree, because the number of children of each node is at most the number of prototiles. By K\H{o}nig's lemma, it has an infinite path, which describes a consistent choice of prototile for each tile of the chosen half-plane binary tiling.
\end{proof}

Thus, in the following undecidability result of Kari, we need not distinguish between the hyperbolic and square binary tilings.

\begin{lemma}[Kari~\cite{Kar-MCU-07}]
The binary Wang tiling problem is undecidable.
\end{lemma}

\begin{theorem}
\label{thm:wang}
The problem of determining the flat-foldability of a crease pattern for a square of origami paper in which each tile of the square binary tiling has a scaled copy of the same finite set of (labeled or unlabeled) creases is undecidable.
\end{theorem}

\begin{proof}
This follows by converting any hard instance of binary Wang tiling into a Boolean circuit for each tile that shares wires with its neighboring tiles describing the color of their shared edges, and checks that this coloring matches a prototile. We convert this circuit into a finite crease pattern, as in \cref{thm:undecidable}, that can be flat-folded within each tile if this check succeeds, and fails to flat-fold otherwise. The crease pattern that repeats this finite pattern for each tile of the binary tiling is flat-foldable if and only if the given Wang tiling instance has a solution.
\end{proof}

\bibliographystyle{plainurl}
\bibliography{folding}
\end{document}